\documentclass[12pt,a4paper]{amsart}
\usepackage[utf8]{inputenc}
\usepackage[T1]{fontenc}
\usepackage{lmodern}
\usepackage[english]{babel}
\usepackage{amsmath,mathtools,amsthm,amssymb,amsfonts}
\usepackage{color}
\usepackage{geometry}
\usepackage{mathrsfs}
\usepackage{graphicx}
\usepackage{enumitem}
\usepackage{tikz} 
\usepackage{pgfplots}
\usepackage{esvect}
\usepackage{pdfpages}
\usetikzlibrary{shapes,arrows}
\usetikzlibrary{calc}
\usetikzlibrary{shapes.geometric}
\usetikzlibrary{shapes.arrows}
\usetikzlibrary{arrows.meta}
\usetikzlibrary{decorations.markings}
\usetikzlibrary{fit}
\usetikzlibrary{patterns}
\usetikzlibrary{hobby}
\usepgfplotslibrary{patchplots}
\pgfplotsset{compat=1.11}
\usepackage{hyperref}
\usepackage{nicefrac}
\usepackage{cancel}
\usepackage{parskip} 

\theoremstyle{plain}
\newtheorem{theorem}{Theorem}[section]
\newtheorem{corollary}[theorem]{Corollary}

\newtheorem{proposition}[theorem]{Proposition}
\newtheorem{lemma}[theorem]{Lemma}
\theoremstyle{definition}
\newtheorem{definition}[theorem]{Definition}
\newtheorem{remark}[theorem]{Remark}
\newtheorem{example}[theorem]{Example}
\newtheorem*{theorem*}{Theorem}
\newtheorem*{definition*}{Definition}
\newtheorem*{corollary*}{Corollary}
\newtheorem*{remark*}{Remark}
\newtheorem*{thm*}{Theorem}
\newtheorem*{conjecture*}{Conjecture}

\newcommand{\definedas}{\mathrel{\raise.095ex\hbox{\rm :}\mkern-5.2mu=}}
\newcommand{\asdefined}{\mathrel{=\mkern-5.2mu\raise.095ex\hbox{\rm :}}}

\newcommand{\R}{\mathbb{R}}

\newcommand{\diver}{\operatorname{div}}
\newcommand{\tr}{\operatorname{tr}}

\newcounter{flabelcounter}
\allowdisplaybreaks

\title[]{Rigidity and Positivity of Hawking Quasi-Local Energy on Area-Constrained Critical Surfaces}
\author[Pe\~nuela Diaz]{Alejandro Pe\~nuela Diaz}
\address{ University of Potsdam, 14476 Potsdam, Germany}
\email{alejandro.penuela.diaz@uni-potsdam.de}

\begin{document}
\begin{abstract}
 A key test for any quasi-local energy in general relativity is that it be nonnegative and satisfy a rigidity property; if it vanishes, the region enclosed is flat.  We show that the Hawking energy, also known as the Hawking mass, satisfies these properties under the dominant energy condition when evaluated on its natural area-constrained critical surfaces within a spacelike hypersurface (initial data set).  In the time-symmetric case, these critical surfaces coincide with area-constrained Willmore surfaces, and we obtain positivity and rigidity theorems for the Hawking energy on such surfaces, including charged and   cosmological constant (hyperbolic and spherical) variants as well as higher-dimensional analogues. In the fully dynamical (non-time-symmetric) case, we establish the first nonnegativity and rigidity theorems for the Hawking energy in this general setting. These results confirm the Hawking energy's consistency with basic physical principles and address several longstanding ambiguities and criticisms.
\end{abstract}
\maketitle

\section{Introduction and Results}
One of the longstanding challenges in classical general relativity is the search for a robust quasi-local energy definition. That is to assign to each finite region of spacetime a physically meaningful notion of energy or mass. While the ADM and Bondi masses capture total energy for isolated systems at spatial or null infinity, there is no unique ``quasi-local’’ analogue measuring the energy contained inside an arbitrary closed 2-surface. Over the decades, many candidates have been proposed, each with its own advantages and limitations (see \cite{Living} for a comprehensive review).  To be considered viable, these definitions must satisfy certain physical conditions. In this paper, we will focus on the following two fundamental conditions under the dominant energy condition:

\begin{enumerate}
    \item  Positivity (i.e. nonnegativity): The energy measure must always be nonnegative.

     \item Rigidity: The energy measure should vanish if and only if the enclosed region is flat. This ensures that quasi-local energy distinguishes between flat and curved spacetimes. 
\end{enumerate}

Among the quasi-local energy candidates, one of the most well-known is the Hawking energy. Introduced by Hawking in 1968 in his pioneering work \cite{Hawma}, this quasi-local energy arises from his study of gravitational radiation as perturbations in an expanding FLRW spacetime. He proposed a quasi-local quantity to measure the total mass  enclosed by a given closed spacelike 2-surface $\Sigma$ and is designed to decrease monotonically as gravitational radiation is emitted.  This quantity, now commonly referred to as the \emph{Hawking energy} or \emph{Hawking mass}, provides a measure of the gravitational energy enclosed by $\Sigma$ in terms of the focusing properties of light rays passing through $\Sigma$, as quantified by the null expansions. The \emph{Hawking energy} \( \mathcal{E}(\Sigma) \) of a closed spacelike 2-surface \( \Sigma \) is given by
\begin{equation}\label{Hawkingene}
     \mathcal{E}(\Sigma) = \sqrt{\frac{|\Sigma|}{16 \pi}} \left( 1+ \frac{1}{16 \pi} \int_\Sigma \theta^+ \theta^- d\mu \right),
\end{equation} 
where $|\Sigma|$ is the area of $\Sigma$, and $\theta^+$ and $\theta^-$ are the null expansions with respect to future-directed null normals $\ell^+$ and $\ell^-$ normalized by $\langle \ell^+,\ell^-\rangle=-2$.

This definition highlights that the Hawking energy measures energy in terms of the bending of light rays on \( \Sigma \), as expressed through the null expansions \( \theta^\pm \). If \( \Sigma \) is the outer boundary of a spacelike hypersurface \( \Omega \), then \(  \mathcal{E}(\Sigma) \) can be interpreted as the total energy enclosed within \( \Sigma \) on \( \Omega \).

The Hawking mass is arguably the simplest proposal for measuring energy in a bounded region, and it satisfies many of the desirable properties (e.g.\ the ADM limit, the small-sphere limit, and monotonicity under inverse-mean-curvature flow). However, it is not positive in general: in Euclidean space every non-round sphere has strictly negative mass, and the only 2-sphere with nonnegative mass is the round sphere (of zero mass). This highlights the need to identify special surfaces on which positivity can be regained.

Christodoulou and Yau were the first to single out the importance of evaluating the Hawking energy in appropriate surfaces and in \cite{Chriyau} they showed that under the dominant energy condition, the Hawking energy is nonnegative on stable constant mean curvature (CMC) spheres in the time-symmetric case.  Shi, Wang,  and Wu  \cite{shi2009behavior} and later Miao, Wang, and Xie \cite{Miao} showed that the Hawking energy converges to the  ADM energy at infinity when evaluated on CMC spheres. More recently, Sun \cite{sun2017rigidity} established that the Hawking energy satisfies rigidity properties on stable CMC spheres. To date, all these results are confined to the time-symmetric case; establishing analogous properties in the fully dynamical setting has proved more elusive.

To overcome this restriction, one may instead seek local maximizers of the Hawking energy, leading naturally to the study of its area-constrained critical surfaces.

We will work in the initial data set setting, this means that we consider a smooth $3$-dimensional Riemannian manifold $(M,g)$, which will be equipped with a symmetric $2$-tensor $k$, and we denote this manifold as a triple $(M,g,k)$.   In this setting, the Hawking energy can be written for a surface $\Sigma \subset M $ as
\begin{equation}\label{hawkingmass2}
    \mathcal{E}(\Sigma) = \sqrt{\frac{|\Sigma|}{16 \pi}} \left( 1- \frac{1}{16 \pi} \int_\Sigma H^2 - P^2 d\mu \right),
\end{equation}
where $H$ is the mean curvature of the surface $\Sigma$ and   $P=\tr_{g_{\Sigma}} k$ is the trace of the tensor $k$ with respect to the metric induced in $\Sigma$, that is  $P= \tr_\Sigma k= \tr k -k(\nu,\nu)$, where $\nu$ is the outward  normal to $\Sigma$ in $M$. 

From a variational point of view,  studying (\ref{hawkingmass2}) for a fixed area is equivalent to studying the Hawking functional 
\begin{equation}\label{hawfun}
    \mathcal{H}(\Sigma)= \frac{1}{4} \int_\Sigma H^2 - P^2 d\mu.
\end{equation}
We are going to consider area-constrained critical surfaces of this functional. In case $k=0$, in a totally geodesic hypersurface, the Hawking functional reduces to the Willmore functional 
\begin{equation}\label{willfunc}
    \mathcal{W}(\Sigma)= \frac{1}{4} \int_\Sigma H^2  d\mu 
\end{equation}
and the critical surfaces of this functional subject to the constraint that $|\Sigma| $ be fixed are the area-constrained Willmore surfaces which we call here for simplicity just \emph{Willmore surfaces}. These surfaces are characterized by the following Euler Lagrange equation with the Lagrange parameter $\lambda$.
\begin{equation}\label{Willeq}
    0= \lambda H  +\Delta^\Sigma H + H|\mathring{B}|^2+ H \mathrm{Ric}^M(\nu, \nu), 
\end{equation}
where $\mathring{B}$   is the traceless part of the second
fundamental form $B$ of $\Sigma$ in $M$, that is $ \mathring{B} = B- \frac{1}{2} H g_\Sigma$ with norm  $|\mathring{B}|^2 = \mathring{B}_{ij}\, g_\Sigma^{ip}\, g_\Sigma^{jq}\, \mathring{B}_{pq}$,  $\mathrm{Ric}^M$ is the Ricci curvature of $M$, $\nu$ is the outward normal to $\Sigma$  and $\Delta^\Sigma $ is the Laplace-Beltrami operator on $\Sigma$.

Willmore surfaces, which have been the focus of extensive mathematical study, were first introduced in the context of general relativity by Lamm, Metzger, and Schulze in \cite{willflat}. They proved the existence of a unique foliation by Willmore spheres in asymptotically flat manifolds. This foliation, known as a foliation at infinity, covers the entire manifold except for a compact region. In their analysis, they proposed that Willmore surfaces are the optimal choice for evaluating the Hawking energy, particularly in manifolds with nonnegative scalar curvature. This claim is substantiated by two fundamental results.

$\bullet$ The Hawking energy is nonnegative on these surfaces.

$\bullet$ The Hawking energy is monotonically nondecreasing along the foliation. 

It was also shown in \cite{Thomas} by Koerber that the leaves of the foliation are strict local area preserving maximizers of the Hawking energy.

There are several results regarding the nonnegativity and monotonicity of the Hawking energy on Willmore and CMC surfaces. However, fewer studies address the rigidity of the Hawking energy, specifically the conditions under which vanishing Hawking energy implies that the domain enclosed by the surface is flat. This property is crucial for the physical viability of any quasi-local energy since it shows that it distinguishes between flat and curved spaces.

In the dynamical setting ($k\neq 0$), a natural class of test‐surfaces are the area-constrained critical surfaces of the Hawking functional (\ref{hawfun}).   
\begin{definition*}
 We call \emph{Hawking surfaces} the area-constrained critical surfaces of the Hawking functional $ \int_\Sigma H^2 -P^2 d\mu$. These surfaces are characterized by the equation 
 \begin{equation*}
\begin{split}
   0=& \lambda H  +\Delta^\Sigma H + H|\mathring{B}|^2+ H \mathrm{Ric}^M(\nu, \nu)+P( \nabla_\nu \tr k - \nabla_\nu k(\nu,\nu )) - 2P \diver_\Sigma (k(\cdot, \nu))\\ & +\frac{1}{2}H P^2 - 2k (\nabla^\Sigma P, \nu )
    \end{split}
\end{equation*}
for some real parameter $\lambda$.
\end{definition*}
These surfaces were already studied in \cite{Alex, diaz2023local}, where the small sphere limit of the Hawking energy for such surfaces was studied. In this paper, we will see that these surfaces are  particularly well adapted for the Hawking energy, in particular  we present the first rigidity results for the Hawking energy in the dynamical setting, alongside results on positivity.

\subsection{Organization of the paper.}
In Section \ref{sectiontime} we focus on the time-symmetric setting. We prove nonnegativity and rigidity results for the Hawking energy on area-constrained Willmore surfaces in several settings: the standard case, the charged case, the positive and negative cosmological constant cases, corresponding respectively to the spherical and hyperbolic reference geometries, and higher-dimensional analogues.

In Section \ref{sectiondyn} we turn to the fully dynamical regime ($k\neq 0$), introducing area-constrained critical surfaces of the Hawking functional and proving the first nonnegativity and rigidity results for the Hawking energy in this general setting. We also discuss the technical sign assumptions appearing in the dynamical results and provide examples illustrating their limitations.

\subsection{Main results time-symmetric setting (\texorpdfstring{$k=0$}))}

Our first result is the rigidity on  Willmore surfaces.

\noindent\textbf{Theorem \ref{rigiditywillmore}.} \emph{Let $(M,g)$ be a $3$-dimensional Riemannian manifold with nonnegative scalar curvature,  and let  $\Omega \subset M$ be a relatively compact domain whose smooth boundary  $\partial \Omega $  has finitely many components, each with positive mean curvature. Suppose that one of the components $\Sigma$, is an area-constrained Willmore surface  and nonnegative Lagrange parameter, that is, it  satisfies
    \begin{equation*}
        0= \lambda H  +\Delta^{\Sigma} H + H|\mathring{B}|^2+ H \mathrm{Ric}^M(\nu, \nu)
    \end{equation*}
   for $\lambda \geq 0$, and the rest of components have positive scalar curvature.  If  $\int_{\Sigma } H^2 d\mu =16 \pi   $ (its Hawking energy is zero)  then   $ \partial \Omega$ is connected and  isometric to a round sphere, and $\Omega$ is isometric to a Euclidean ball in $\mathbb{R}^3$.}

  In particular, with this, we can deduce a positive mass theorem for the Hawking energy.

  \noindent\textbf{Corollary \ref{positivemasshaw}.} \emph{Let $(M,g)$ be a $3$-dimensional Riemannian  manifold with nonnegative scalar curvature. Suppose $\Omega$ is a relatively compact domain with smooth connected boundary  $\Sigma=  \partial \Omega $.  Let $\Sigma $ be an  area-constrained Willmore surface with positive mean curvature and nonnegative Lagrange parameter, then the Hawking energy satisfies
  $$  \mathcal{E}(\Sigma) = \sqrt{\frac{|\Sigma|}{16 \pi}} \left( 1- \frac{1}{16 \pi} \int_\Sigma H^2  d\mu \right) \geq 0 $$
  with equality if and only if $\Omega$ is isometric to a Euclidean ball in $\mathbb{R}^3$ and $\Sigma$ is isometric to a round sphere. }

 We extend these rigidity and nonnegativity properties to encompass several major physical and geometric generalizations:
\begin{itemize}
    \item \textbf{Electrically Charged Manifolds (Corollary \ref{coroelec}):} We recover positivity and rigidity under the Einstein-Maxwell dominant energy condition.
     \item \textbf{Cosmological backgrounds (Theorems \ref{rigiditywillhyper} and \ref{positivecosmo}):} We establish nonnegativity and rigidity in the hyperbolic reference geometry, corresponding to negative cosmological constant $\Lambda<0$, under the scalar curvature bound $\mathrm{Sc}^M\ge 2\Lambda$. In the spherical reference geometry, corresponding to positive cosmological constant $\Lambda>0$, we prove nonnegativity under the same scalar curvature bound and obtain rigidity under the stronger ambient Ricci lower bound $\mathrm{Ric}^M\ge \frac{2}{3}\Lambda g$.

    \item \textbf{Higher Dimensions (Section \ref{higherdimsec}):} We introduce higher-dimensional Hawking-type energies and prove corresponding positivity and rigidity theorems in the Euclidean and spherical reference settings for area-constrained Willmore hypersurfaces.
\end{itemize}
 We also obtain a strong rigidity result for the foliation of Willmore surfaces. 

\noindent\textbf{Theorem \ref{rigidityfoli}.} \emph{Let $(M, g)$ be a complete $3$-dimensional asymptotically flat  Riemannian manifold with nonnegative scalar curvature. Then the Hawking and  Brown-York energy of all the Willmore surfaces of the canonical Willmore foliation are positive unless $(M,g)$ is isometric to Euclidean space.}

\subsection{Main results dynamical setting (\texorpdfstring{$k\neq 0$}))}

In this setting, we will focus on the Hawking surfaces introduced before. Note that these surfaces are defined within a given spacelike hypersurface. Consequently, defining them independently of a specific hypersurface would require selecting a preferred spacelike normal direction for variation, introducing an inherent gauge dependence into the definition.

We obtain the following result, which shows the positivity and rigidity of the Hawking energy on such surfaces under a technical assumption.

\noindent\textbf{Theorem \ref{positivity0}.} \emph{Let $(M,g,k)$ be a $3$-dimensional initial data set satisfying the dominant energy condition.}
    
    \emph{$i)$ Let $\Sigma$ be a Hawking surface with positive mean curvature, and such that for} 
    \begin{equation*}
        f:= \left( \frac{P}{H}\right)^2|k|^2+ \frac{1}{2 }(\tr k)^2    - \frac{3}{4} P^2- \frac{P}{H}( \nabla_\nu \tr k - \nabla_\nu k(\nu,\nu )) -  \frac{1}{2} |k|^2  -\frac{1}{2} |\mathring{B}|^2 -|J|
    \end{equation*}
    \emph{the surface satisfies $\int_\Sigma f -\lambda d\mu \leq 0$. Then $\int_\Sigma H^2 -P^2 d\mu \leq 16\pi $, and  if $\int_\Sigma f -\lambda d\mu < 0$ then  $\int_\Sigma H^2 -P^2 d\mu < 16\pi $. In particular, the Hawking energy is nonnegative.}

 \emph{$ii)$ Let  $\Omega \subset M$ be a relatively compact domain whose smooth boundary  $\partial \Omega $  has finitely many components. Suppose that one of the boundary components $\Sigma$ is a Hawking surface with positive mean curvature, and the other components have positive scalar curvature and spacelike mean curvature vector ($H^2-P^2>0$). If   there exists a constant $0\leq\beta <\frac{1}{2}$ such that $\int_\Sigma f_\beta -\lambda \, d\mu \leq 0$  for} 
 \begin{equation*}
        f_\beta:= \left( \frac{P}{H}\right)^2|k|^2+ \frac{1}{2 }(\tr k)^2   - \frac{3}{4} P^2- \frac{P}{H}( \nabla_\nu \tr k - \nabla_\nu k(\nu,\nu ))  -  \beta( |k|^2 + |\mathring{B}|^2 +2|J|),
    \end{equation*}
\emph{and  $\int_\Sigma H^2 -P^2 d\mu = 16\pi $. Then   $\Omega$ is isometric to a spacelike hypersurface in Minkowski spacetime with second fundamental form  $k$,   $\partial \Omega$ is connected ($ \partial \Omega=\Sigma $), umbilic, isometric to a round sphere and $k=0$ on $\Sigma$.}

 Note that the condition \( \int_\Sigma f_\beta -\lambda \, d\mu \leq 0 \) is a strengthening of \( \int_\Sigma f -\lambda \, d\mu \leq 0 \). Neither of these conditions is optimal nor physically motivated. In particular, the function $f_\beta$ was introduced for a purely technical reason. In Remark \ref{remarf} we will see that one can also define an alternative $f$ given by  $$\Tilde{f}:= \frac{2P}{H}  k(\nabla^\Sigma \log H, \nu)+ \frac{1}{2 }(\tr k)^2    - \frac{3}{4} P^2- \frac{P}{H}( \nabla_\nu \tr k - \nabla_\nu k(\nu,\nu )) -  \frac{1}{2} |k|^2  -\frac{1}{2} |\mathring{B}|^2 -|J|,$$ and the same result would hold.  We will also see that the condition  on \( f_\beta  \) might be artificially enforcing the Hawking energy to be too positive, as it is seen in the following result.

\noindent\textbf{Corollary \ref{positivemasshawdyn}.} \emph{Let $(M,g,k)$ be a $3$-dimensional compact hypersurface in Minkowski spacetime. Assume that the boundary of $M $, $\partial M =\Sigma $   is a Hawking surface of positive mean curvature and that there exists a constant $\beta <\frac{1}{2}$ such that $\int_\Sigma f_\beta -\lambda \, d\mu \leq 0$. Then the Hawking energy on $\Sigma$ is strictly positive unless $\Sigma$ is also  contained in a hyperplane of Minkowski spacetime.}

The excess positivity of the Hawking energy could stem from one of two factors. The first possibility is that the technical condition \( \int_\Sigma f_\beta -\lambda \, d\mu \leq 0 \) imposes an overly restrictive constraint, biasing the selection of Hawking surfaces toward those with higher energy. The second possibility is that the Hawking energy on Hawking surfaces is inherently 'too positive,' meaning that these surfaces introduce an excess contribution to the energy measurement.  We are more inclined for the too restrictive technical condition, as it is illustrated in Examples \ref{example1} and \ref{example2}.

Finally, we show the nonnegativity of the Hawking energy in higher dimensions when evaluated on Hawking surfaces in Theorem \ref{posidynhigh}.

We have shown that, on the appropriate area-constrained surfaces, the Hawking mass is nonnegative and rigid. However, in the fully dynamical case ($k\neq0$), our rigidity hypothesis appears stronger than necessary, and it may be biasing the Hawking mass toward excessive positivity.

Some of the results proved here, together with further applications of Hawking surfaces, first appeared in the author’s PhD thesis \cite{PenuelaThesis}. The construction of foliations at infinity, monotonicity and large-sphere limits of the Hawking energy along these foliations, and verification of the hypotheses on $f$ for a wide class of initial data are developed in the companion paper \cite{penuelafol}.

With these results, we show that the Hawking mass, when evaluated on spacelike Hawking surfaces, satisfies the key physical and mathematical criteria expected of a quasi-local energy measure. In particular, we aim to elevate the Hawking energy’s status as a viable and useful quasi-local energy measure under realistic conditions.

\section{Time-symmetric setting (\texorpdfstring{$k=0$}))}\label{sectiontime}
In this section, we will mostly work on domains in a manifold, i.e., on connected open sets.

The nonnegativity of the Hawking energy evaluated on Willmore surfaces was proved by Lamm, Metzger, and Schulze:
\begin{theorem}[{\cite[Theorem 4]{willflat}}]\label{theo4willfla}
    If \((M, g)\) satisfies \( \mathrm{Sc}^M \geq 0 \) and if \( \Sigma \) is a compact spherical area-constrained Willmore surface  with \( H > 0 \), then \( \mathcal{E} (\Sigma) \geq 0 \) if \( \lambda \geq 0 \).
\end{theorem}

\begin{remark}
   A closer look at the proof shows that the exact condition needed for nonnegative Hawking energy is $ \int_{\Sigma}  \lambda+|\nabla^{\Sigma}  \log H|^2  +\frac{1}{2} |\mathring{B}|^2  \, d\mu \geq 0 $ and in particular this is automatic whenever $\lambda \geq 0$.
\end{remark}
As mentioned above, rigidity results for the Hawking energy are still rather scarce.  A common approach is to look for {\em unconstrained} local maximizers  of the Hawking energy, as in   \cite{baltazar2023local,barros2017hawking, lee2025modified, maximo2012hawking, sousa2023charged}.  In each of these works one assumes a lower bound on the scalar curvature of the ambient manifold and shows that, if there exists a minimal surface \(\Sigma\) which locally maximizes the Hawking energy (adjusted to include a cosmological constant term), then a neighbourhood of \(\Sigma\) must be isometric to one of the standard black-hole models with cosmological constant (Schwarzschild–de Sitter, Reissner–Nordström–de Sitter or anti–de Sitter).  Note, however, that these \(\Sigma\) are {\em not} area‐constrained critical surfaces, and the resulting rigidity statements concern local Schwarzschild–(A)dS geometry rather than flatness of the enclosed region.  Also, recent work \cite{leandro2025sharp} establishes that, in electrostatic manifolds, attaining a sharp lower bound for the charged Hawking on a minimal surface energy forces the surface to coincide with the Reissner–Nordström–de Sitter horizon.

A more pointwise rigidity theorem is due to Mondino and Templeton‐Browne \cite{mondino2022some}: they show that if an open set \(\Omega\subset M\) has the property that, at every point \(p\in\Omega\), there is a neighbourhood \(U\subset\Omega\) in which the supremum of the Hawking energy of all surfaces contained in \(U\) is nonpositive, then \(\Omega\) is locally isometric to Euclidean \( \R^3\) (resp.\ to hyperbolic space \( \mathbb{H}^3\)).  While this result is closer in spirit to our flat‐interior rigidity, its hypothesis is very strong, requiring a uniform energy bound in every sufficiently small ball.

The most successful flatness rigidity results so far have been obtained for stable constant mean curvature surfaces in the time-symmetric setting, with the main result being a combination of the main results in \cite{shi2019uniqueness, sun2017rigidity}, which we state as:
\begin{theorem}[{\cite[Theorem 2, Theorem 1]{shi2019uniqueness, sun2017rigidity}}]
  Let \((M, g)\) be a $3$-dimensional Riemannian manifold with nonnegative scalar curvature, and let \( \Omega \subset M \) be a relatively  compact domain with smooth boundary \( \Sigma = \partial \Omega \). If \( \Sigma \) is a stable constant mean curvature  sphere with vanishing Hawking energy ( $\int_{\Sigma } H^2 d\mu =16 \pi   $),  either $\Sigma$ has  even symmetry, or its Gauss curvature  \( K_{\Sigma} \) is \( \mathcal{C}^0 \)-close to \( \frac{4\pi}{|\Sigma|} \), i.e. either  there exist an isometry $\rho :\Sigma \to \Sigma $ with $\rho^2 =id$ and $\rho(x)\neq x$ for $x \in \Sigma$ or \(
|K_{\Sigma} -\frac{4\pi}{|\Sigma|}|_{\mathcal{C}^0} < \delta_0
\) for some \( \delta_0 \ll 1 \).
 Then \( \Omega \) is isometric to a Euclidean ball in \( \mathbb{R}^3 \). In particular, \( \Sigma \) is isometric to the round sphere in \( \mathbb{R}^3 \).
\end{theorem}
This result and the following results on  Willmore surfaces rely on their core  on the rigidity of the Brown-York energy, which is a consequence of the following result of Shi and Tam.
\begin{theorem}[{\cite[Theorem 1]{shi2002positive}}]\label{rigiditybrown}
    Let \((\Omega, g)\) be a compact manifold of dimension three with a smooth boundary and with nonnegative scalar curvature. Suppose \(\partial \Omega\) has finitely many components \(\Sigma_i\) such that each component has positive Gaussian curvature and positive mean curvature \(H^i \) with respect to the unit outward normal. Then for each boundary component \(\Sigma_i\),
\begin{equation}\label{inequbrown}
    \int_{\Sigma_i} H^i \, d\mu \leq \int_{\Sigma_i} H_0^{i} \, d\mu
\end{equation}
where \(H_0^{i}\) is the mean curvature of \(\Sigma_i\) with respect to the outward normal when it is isometrically embedded in \(\mathbb{R}^3\), and \(d\mu \) is the volume form on \(\Sigma_i\) induced from \(g\). Moreover, if equality holds in (\ref{inequbrown}) for some \(\Sigma_i\), then \(\partial \Omega\) has only one component and \(\Omega\) is a domain in \(\mathbb{R}^3\).
\end{theorem}
  Note that the isometric embedding onto Euclidean space is unique because of the positive Gauss curvature of the surfaces. This is thanks to the Weyl-Nirenberg-Pogorelov Theorem  \cite{nirenberg1953weyl, pogorelov1972extrinsic}. 
\begin{theorem}[Weyl-Nirenberg-Pogorelov]\label{WeyNi}
Let $(S^2,g)$ be a $C^{k,\alpha}$ ($k\ge3$, $\alpha\in(0,1)$) Riemannian 2‐sphere with Gaussian curvature $K_g>0$.  Then there exists a strictly convex embedding
\[
X\colon (S^2,g)\;\hookrightarrow\;(\mathbb R^3,g_{\mathrm{Eucl}})
\]
which is an isometry onto its image, and any two such embeddings differ by an orientation‐preserving rigid motion of~$\mathbb R^3$.
\end{theorem}
In \cite{shi2002positive} Shi and Tam also proved a higher dimensional version of Theorem \ref{rigiditybrown}, which we state as follows.
\begin{theorem}[{\cite[Theorem 4.1]{shi2002positive}}]\label{highershitam}
    For \(n \geq 3\), suppose \((M^n, g)\) is a compact manifold with boundary \(\Sigma := \bigcup_{i=1}^m \Sigma_i\), where each \((\Sigma_i, g_{|\Sigma_i})\) is a connected component that can be isometrically embedded in \(\mathbb{R}^n\) as a convex hypersurface. Assume \(3 \leq n \leq 7\) or \(M\) is spin. Moreover, its scalar curvature $
\mathrm{Sc}^M \geq 0 $
and the mean curvature of \(\Sigma_i\) with respect to \(g\) satisfies $
H^i > 0$  on  $\Sigma_i$, then the Brown-York energy
\[
\mathcal{E}_{\text{BY}}(\Sigma_i, g) := \frac{1}{8 \pi }\int_{\Sigma_i} \left(H^i_0 - H^i \right) d\mu \geq 0, \quad i = 1, \dots, m,
\]
where \(H^i_0\) is the mean curvature of \(\Sigma_i\) with respect to the Euclidean metric. Moreover, if one of the energies vanishes, then the boundary has only one component and \((M, g)\) is isometric to a bounded domain in \(\mathbb{R}^n\).
\end{theorem}
\begin{remark}\label{ramarkhierdim} 
Note that the proof of these last two results relies on the positive mass theorem. As claimed by Lohkamp \cite{lohkamp2016higher1, lohkamp2016higher2} and Schoen-Yau \cite{schoen2017positive} independently, the positive mass theorem for dimensions \(n \geq 8\)  is still valid. 
Since the assumptions on dimensions and spin structures in the theorem only serve to ensure the ADM mass’s positivity, they can be omitted. Thus, whenever Theorem \ref{highershitam} is applied, we refer to this improved version. 
\end{remark}
We will now demonstrate that Willmore surfaces are particularly well-suited for establishing rigidity results for the Hawking energy. 
\begin{theorem}\label{rigiditywillmore}
    Let $(M,g)$ be a $3$-dimensional Riemannian manifold with nonnegative scalar curvature,  and let  $\Omega \subset M$ be a relatively compact domain whose smooth boundary  $\partial \Omega $  has finitely many components, each with positive mean curvature. Suppose that one of the components $\Sigma$, is an area-constrained Willmore surface  and nonnegative Lagrange parameter, that is, it  satisfies
    \begin{equation}\label{willmoreequa}
        0= \lambda H  +\Delta^{\Sigma} H + H|\mathring{B}|^2+ H \mathrm{Ric}^M(\nu, \nu)
    \end{equation}
   for $\lambda \geq 0$, and the rest of components have positive scalar curvature.  If  $\int_{\Sigma } H^2 d\mu =16 \pi   $ (its Hawking energy is zero)  then   $ \partial \Omega$ is connected and  isometric to a round sphere, and $\Omega$ is isometric to a Euclidean ball in $\mathbb{R}^3$.   
\end{theorem}
\begin{proof}
 We start by multiplying equation (\ref{willmoreequa}) by \( H^{-1} \) and integrating the first term by parts. This yields:
\[
\lambda |\Sigma | + \int_{\Sigma }  |\nabla^{\Sigma}  \log H|^2 + |\mathring{B}|^2 +  \, \mathrm{Ric}^M(\nu, \nu)  \, d\mu = 0. 
\]
Now by the  Gauss equation 
$$ \mathrm{Sc}^{\Sigma} = \mathrm{Sc}^{M} - 2\mathrm{Ric}^M(\nu, \nu) + \frac{1}{2}H^2 - |\mathring{B}|^2 $$
and the Gauss-Bonnet formula we get
\[
\lambda |\Sigma | + \int_{\Sigma} |\nabla^{\Sigma}  \log H|^2 + \frac{1}{4} H^2 + \frac{1}{2} |\mathring{B}|^2  \, d\mu \leq 4\pi - \int_{\Sigma } \frac{1}{2} \, \mathrm{Sc}^M \, d\mu
\]   
and $ \int_{\Sigma } \frac{H^2}{4} d\mu =4 \pi   $.  This implies $\lambda=|\mathring{B}|= \mathrm{Sc}^M_{| \Sigma} = 0 $, that $H$ is constant and that $\Sigma$ is a sphere since $\int_{\Sigma } \frac{1}{2}\mathrm{Sc}^{\Sigma} \, d\mu = 4 \pi$. Then by (\ref{willmoreequa}) we also have $ \mathrm{Ric}^M(\nu, \nu)=0$ along $\Sigma $.  Now again by the Gauss equation, it is direct to see that $ \mathrm{Sc}^{\Sigma} = \frac{2}{r^2}$ where $r$ is the area radius of $\Sigma$. Now with this, we can apply the rigidity result of Theorem \ref{rigiditybrown}. First note that since the Gauss curvature of $\Sigma$ is  $ \frac{1}{r^2}$, the isometric  embedding  of $\Sigma$ into $\mathbb{R}^3$ is a round sphere, therefore $H_0= \frac{2}{r}$,  by Theorem \ref{rigiditybrown} and its rigidity we have our result. 
\end{proof}
\begin{remark}
 $ i)$ Note that by considering Willmore surfaces we do not need the surface to be a priori a topological sphere or almost round.

 $ii)$ The condition $\lambda \geq 0 $ can be improved to the condition that $ \int_{\Sigma}  \lambda+\alpha|\nabla^{\Sigma}  \log H|^2  + \beta |\mathring{B}|^2  \, d\mu \geq 0 $ holds    for any constants $0<\alpha<1 $, $0<\beta<\frac{1}{2}$. 
\end{remark}

\begin{remark}
In general, the possibility of the Hawking energy being negative is often regarded as a drawback. However, in the study of spaces with zero-area singularities (see. for instance, \cite{bray2011positive, ajm/1383923956}), this feature becomes advantageous. These singularities are associated with spacetimes of negative mass, and the negativity of the Hawking energy provides a useful tool for analyzing them. 
In this context, it is also important to carefully select the surfaces on which the Hawking energy is evaluated. One might expect that, in a manifold with nonpositive scalar curvature, evaluating the Hawking energy on Willmore surfaces with a nonpositive Lagrange parameter would yield a nonpositive value. However, a quick computation reveals that this is not necessarily the case.

\end{remark}

In Euclidean space, spherical Willmore surfaces with $\lambda \geq 0$ are round spheres, which, in particular, have zero Hawking energy. Consequently, the previous result directly yields a positive mass theorem for the Hawking energy on Willmore surfaces.
\begin{corollary}\label{positivemasshaw}
  Let $(M,g)$ be a $3$-dimensional Riemannian  manifold with nonnegative scalar curvature. Suppose $\Omega$ is a relatively compact domain with smooth connected boundary  $\Sigma=  \partial \Omega $.  Let $\Sigma $ be an  area-constrained Willmore surface with positive mean curvature and nonnegative Lagrange parameter, then the Hawking energy satisfies
  $$  \mathcal{E}(\Sigma) = \sqrt{\frac{|\Sigma|}{16 \pi}} \left( 1- \frac{1}{16 \pi} \int_\Sigma H^2  d\mu \right) \geq 0 $$
  with equality if and only if $\Omega$ is isometric to a Euclidean ball in $\mathbb{R}^3$ and $\Sigma$ is isometric to a round sphere.   
\end{corollary}

In the case of an asymptotically flat manifold, we can get stronger results. In this setting, Sun proved the following result for isoperimetric surfaces, that is, surfaces that enclose a given volume with the minimum possible surface area.

\begin{theorem}[{\cite[Theorem 3]{sun2017rigidity}}]\label{afbrownyork}
    Let \((M, g)\) be a complete asymptotically flat three-manifold with scalar curvature \(  \mathrm{Sc}^M \geq 0\). If there exists an isoperimetric surface with vanishing  Hawking energy and  Gauss curvature  \( \mathcal{C}^0 \)-close to \( \frac{4\pi}{|\Sigma|} \), i.e. \(
|K_{\Sigma} -\frac{4\pi}{|\Sigma|}|_{\mathcal{C}^0} < \delta_0
\) for some \( \delta_0 \ll 1 \). Then \((M, g)\) is isometric to \((\mathbb{R}^3, \delta)\), where \(\delta\) denotes the Euclidean metric on \(\mathbb{R}^3\).
\end{theorem}

The result relies on the following result of Shi.
\begin{theorem}[{\cite[Theorem 3]{shi2016isoperimetric}}]\label{isoperig}
    Suppose \((M, g)\) is a complete  asymptotically flat manifold with nonnegative scalar curvature. Then for any  \(V > 0\),
\begin{equation}
I(V) \leq (36\pi)^{\frac{1}{3}} V^{\frac{2}{3}}. 
\end{equation}
There exists a  \(V_0 > 0\) such that
\begin{equation}
I(V_0) = (36\pi)^{\frac{1}{3}} V_0^{\frac{2}{3}} 
\end{equation}
if and only if \((M, g)\) is isometric to \(\mathbb{R}^3\). Here
\[
I(v) = \inf \left\{ \mathcal{H}^2(\partial^* \Omega) : \Omega \subset M \text{ is a Borel set with finite perimeter, and } \mathcal{L}^3(\Omega) = v \right\},
\]
is the isoperimetric profile, where \(\mathcal{H}^2\) is the 2-dimensional Hausdorff measure of the reduced boundary \(\partial^* \Omega\), and \(\mathcal{L}^3(\Omega)\) is the Lebesgue measure of \(\Omega\) with respect to the metric \(g\).
\end{theorem}

It is direct to prove a similar result to Theorem \ref{afbrownyork} but with the Brown-York energy.

\begin{proposition}\label{isobrown}
    Let \((M, g)\) be a complete asymptotically flat  three-manifold with scalar curvature \(\mathrm{Sc}^M \geq 0\). If there exists an isoperimetric surface \(\Sigma\) with positive mean and Gauss curvatures, and vanishing Brown-York energy. Then \((M, g)\) is isometric to \((\mathbb{R}^3, \delta)\), where \(\delta\) denotes the Euclidean metric on \(\mathbb{R}^3\).
\end{proposition}
\begin{proof}
Our surface satisfies $\int_\Sigma H d\mu = \int_\Sigma H_0 d\mu$, where $H_0$ is the mean curvature of the surface when isometrically embedded in Euclidean space. Then by the rigidity of the Brown-York (Theorem \ref{rigiditybrown}) we obtain not only that the domain $\Omega$ enclosed by $\Sigma$ is a Euclidean one,  but also that $H_0= H=\frac{2}{r}$ (and $ \mathrm{Ric}^M=0$ on $\partial \Omega$). Then we have that $\Sigma$ is a constant mean curvature surface in Euclidean space, and therefore by Alexandrov theorem a round sphere. Then as it is also an isoperimetric surface, we have by Theorem \ref{isoperig} the result. 
\end{proof}

With this result and combining a result of Sun proving the positivity of the Hawking energy on the leaves of the canonical CMC foliation, we can obtain as a consequence
\begin{corollary}
     Let \((M, g)\) be a $3$-dimensional asymptotically flat  Riemannian manifold with nonnegative  scalar curvature. Then the Hawking  and  Brown-York energies of all the large enough constant mean curvature (CMC) surfaces of the canonical CMC foliation are positive unless $(M,g)$ is isometric to Euclidean space.
\end{corollary}
\begin{proof}
    The result for the Hawking energy was proved in \cite[Corollary 2]{sun2017rigidity}. For the Brown-York energy, we first need to note that  in an asymptotically flat manifold, large isoperimetric surfaces are precisely the leaves of the canonical foliation of stable constant mean curvature surfaces \cite{eichmair2013large, Eich}. By the estimates derived in the construction of the CMC foliation \cite{Nerz2}, we have  for $\Sigma$ satisfying $H(\Sigma)= \frac{2}{r}$ in the foliation, it holds 
$$ |\mathring{B}|= O(r^{-\frac{3}{2}-\epsilon}), \, \,  |\mathrm{Ric}^M|= O(r^{-\frac{5}{2}-\epsilon}), \,\,|\mathrm{Sc}^M |=O(r^{-3-\epsilon}) \quad \text{on $\Sigma$} $$
Then by the Gauss equation $\mathrm{Sc}^{\Sigma} = \mathrm{Sc}^{M} - 2\mathrm{Ric}^M(\nu, \nu) + \frac{1}{2}H^2 - |\mathring{B}|^2 $, and the nonnegative scalar curvature of $M$, $\Sigma$ has positive Gauss curvature. Then by Proposition \ref{isobrown} we have the result. 
\end{proof}
Now we will consider the Willmore surfaces of the canonical Willmore foliation derived in \cite{eichko, willflat}.
\begin{theorem}\label{rigidityfoli}
    Let \((M, g)\) be a $3$-dimensional asymptotically flat  Riemannian manifold with nonnegative  scalar curvature. Then the Hawking  and  Brown-York energies of all the  Willmore surfaces of the canonical Willmore foliation  are positive unless $(M,g)$ is isometric to Euclidean space.
\end{theorem}
\begin{proof}
    Assume the contrary for the Hawking energy. By Theorem \ref{theo4willfla}, it is established that the Hawking energy on such surfaces is nonnegative. Thus, there exists a surface $\Sigma$ in the foliation with zero Hawking energy. From Theorem \ref{rigiditywillmore}, it follows that $\Sigma$ is isometric to a round sphere, and its enclosed region is Euclidean. In particular, $\Sigma$ is a stable CMC surface, as round spheres are known to be stable CMC surfaces.  Now by the uniqueness of the canonical CMC foliation \cite{Nerz2}, $\Sigma $ belongs to the foliation. Note that the CMC foliation is in particular a foliation of isoperimetric surfaces, then we have that $\Sigma$ is an isoperimetric surface with the isoperimetric ratio of a Euclidean sphere, then by Theorem \ref{isoperig} we have the result. 

    Now suppose the opposite for the Brown-York energy. As this energy is nonnegative for manifolds with nonnegative scalar curvature, it must be zero on a $\Sigma$ belonging to the foliation. We will see that $\Sigma$ has in particular zero Hawking energy.   By the rigidity result of the Brown-York energy (Theorem \ref{rigiditybrown}),   the domain $\Omega \subset M$ enclosed by $\Sigma$ is isometric to a domain in Euclidean \(\mathbb{R}^3\).  Moreover, along \(\Sigma\) the mean curvatures \(
H = H_0\)
agree with their Euclidean value \(H_0\), and the ambient Ricci curvature vanishes:
\(
\mathrm{Ric}^M\bigl|_{\Sigma} = 0.
\)  Then since $\Omega$ is a domain in $\mathbb{R}^3$, $\frac{1}{4} \int_{\Sigma} H^2 d\mu \geq 4 \pi$, and as $\Sigma$ is a Willmore sphere  we have 
    
    \[
\lambda |\Sigma | + \frac{1}{4} \int_{\Sigma} H^2d\mu-4\pi +\int_{\Sigma} |\nabla^{\Sigma}  \log H|^2  + \frac{1}{2} |\mathring{B}|^2  \, d\mu =  - \int_{\Sigma } \frac{1}{2} \, \mathrm{Sc}^M \, d\mu.
\]   
then since $\mathrm{Sc}^M \geq 0$ we have $\lambda= \frac{1}{4} \int_{\Sigma} H^2 d\mu-4\pi=|\mathring{B}|=\mathrm{Sc}^M_{\Sigma} =0$. Then the surface has zero Hawking energy and the result follows by the first part of the proof.  
\end{proof}
A similar rigidity was proved by Wei \cite[Theorem 1.2]{wei2021minimizers}: if an asymptotically flat 3‐manifold of nonnegative scalar curvature admits a {\em global} area–constrained minimizer of the Willmore functional  with vanishing Hawking energy, then the manifold is flat.  In contrast, Theorem \ref{rigidityfoli} requires only that the Hawking energy vanish on a single Willmore surface, without assuming it attains a global minimum.
\subsection{Charged case}

It is direct to see that the rigidity result also holds for the electrically charged case, first we need to introduce the main concepts of this setting. 

\begin{definition}\label{electrichar}
   A   time-symmetric initial data for the Einstein-Maxwell equations  $ (M,g,E)$  is a Riemannian manifold $(M,g)$ equipped with an electric vector field $E$ which satisfies $ \diver E =4 \pi \rho$, where $ \rho $  is the electric charge density of the matter. In this case, the dominant energy condition reduces to $ \mathrm{Sc}^M \geq 2 |E|^2$.

For a closed surface $\Sigma$ in $M$, we define the charge enclosed by $\Sigma$ to be given by the flux
integral
\begin{equation}\label{charge}
   Q(\Sigma) = \frac{1}{4 \pi} \int_\Sigma g(E,\nu) d\mu 
\end{equation} 
where $\nu$ is the normal to $\Sigma$. In this context, we have that the charged Hawking energy is given by
\begin{equation} 
\mathcal{E}_{Q}(\Sigma) = \sqrt{\frac{|\Sigma|}{16\pi}} 
\left( 
1 + \frac{4\pi Q(\Sigma)^2}{|\Sigma|} 
- \frac{1}{16\pi} \int_{\Sigma} H^2 \, d\mu 
\right),
\end{equation}
\end{definition}
For more details on this definition see \cite{disconzi2012penrose, mccormick2019charged}.

Now, since the charged Hawking energy is larger than the standard one then we have the rigidity for Willmore surfaces as a direct result of Theorem
\ref{rigiditywillmore}.
\begin{corollary}\label{coroelec}
  Let $(M,g,E)$ be a $3$-dimensional time-symmetric initial data for the Einstein-Maxwell equations  which satisfies the dominant energy condition $ \mathrm{Sc}^M \geq 2 |E|^2$,   and let  $\Omega \subset M$ be a relatively compact domain whose smooth boundary  $\partial \Omega $  has finitely many components, each with positive mean curvature.  Suppose that one of the components $\Sigma $ is an area-constrained Willmore surface with  nonnegative Lagrange parameter, that is, it  satisfies
    \begin{equation}
        0= \lambda H  +\Delta^{\Sigma} H + H|\mathring{B}|^2+ H \mathrm{Ric}^M(\nu, \nu)
    \end{equation}
   for $\lambda \geq 0$, and the rest of the components have positive scalar curvature.  If the charged Hawking energy is zero on $\Sigma$, then $\partial \Omega$ is connected and isometric to a round sphere, $\Omega$ is isometric to a Euclidean ball in $\mathbb{R}^3$, and $ E$ vanishes  on $\Omega$.  
\end{corollary}
\begin{remark}\label{chargeremark}
In this work, we have avoided introducing magnetic fields. However, incorporating a magnetic field involves considering an additional vector field $\mathcal{B}$, referred to as the magnetic vector field, which satisfies $\diver \mathcal{B} =\rho_\mathcal{B}$. Typically, 
$\rho_\mathcal{B} $ is set to zero. Under these conditions, the dominant energy condition in a totally geodesic slice ($k=0$) takes the form $ \mathrm{Sc}^M \geq 2 |E|^2 +2|\mathcal{B}|^2$.  In this case, we define the magnetic charge of a
surface $Q_\mathcal{B}(\Sigma)$ in an analogous to (\ref{charge}) and the charged Hawking energy is given by 
$$\mathcal{E}_{Q,Q_B}(\Sigma) = \sqrt{\frac{|\Sigma|}{16\pi}} 
\left( 
1 + \frac{4\pi (Q(\Sigma)^2 +Q_\mathcal{B}(\Sigma)^2 )}{|\Sigma|} 
- \frac{1}{16\pi} \int_{\Sigma} H^2 \, d\mu 
\right). $$
For   more details see \cite[Appendix A]{mccormick2019charged}. When considering this setting  the rigidity result would follow in the same way as before as a direct consequence of Theorem \ref{rigiditywillmore}.
\end{remark}
\begin{remark}
    Note that because of the dependence of  the charge $Q(\Sigma)$ on the surface, the Willmore surfaces are not critical surfaces of the charged  Hawking energy. However if  the electric charge density $\rho$ is zero, then $\diver E =0$ and   $Q(\Sigma)$ is constant for every variation. In this case the Willmore surfaces are critical surfaces of the 
charged  Hawking energy.
\end{remark}
\subsection{Cosmological constant case}

When considering an initial data set of the Einstein equations with cosmological constant $\Lambda $, the dominant energy condition reduces to $\mathrm{Sc}^M \geq 2 \Lambda $. In this setting, one defines the Hawking energy with cosmological constant $\Lambda$,  by   
\begin{equation} 
\mathcal{E}_{\Lambda}(\Sigma) = \sqrt{\frac{|\Sigma|}{16\pi}} 
\left( 
1   
- \frac{1}{16\pi} \int_{\Sigma} H^2+ \frac{4}{3} \Lambda \, d\mu 
\right).
\end{equation}
 Note that when \(\Lambda=0\) it reduces to the usual Hawking energy.  Also that compared to the charged case,  Willmore surfaces are area-constrained critical surfaces of the Hawking energy with cosmological constant $\Lambda$.

We begin with the hyperbolic case ($\Lambda < 0 $). In this setting, the natural “zero-energy” model to compare for rigidity is the  hyperbolic space of constant sectional curvature $\Lambda/3$, denoted by $\mathbb{H}^3_{\Lambda/3}$.

In  \cite{shi2019uniqueness, sun2017rigidity} the rigidity of  the Hawking energy in the hyperbolic case was considered, obtaining the result.
\begin{theorem}[{\cite[Theorem 3, Theorem 2]{shi2019uniqueness, sun2017rigidity}}]
     Let \((M, g)\) be a $3$-dimensional  Riemannian manifold with scalar curvature \(\mathrm{Sc}^M \geq -6 \), and let \( \Omega \subset M \) be a relatively compact domain with smooth boundary \( \Sigma = \partial \Omega \). If \( \Sigma \) is a stable  constant mean curvature  sphere with $\int_{\Sigma } H^2 -4d\mu =16 \pi   $, if either $\Sigma$ has  even symmetry, or its Gauss curvature  \( K_{\Sigma} \) is \( \mathcal{C}^0 \)-close to \( \frac{4\pi}{|\Sigma|} \), i.e. either  there exist an isometry $\rho :\Sigma \to \Sigma $ with $\rho^2 =id$ and $\rho(x)\neq x$ for $x \in \Sigma$ or \(
|K_{\Sigma} -\frac{4\pi}{|\Sigma|}|_{\mathcal{C}^0} < \delta_0
\) for some \( \delta_0 \ll 1 \).
 Then \( \Omega \) is isometric to a hyperbolic  ball in \( \mathbb{H}^3 \).
\end{theorem}
For the next result, the rigidity of the Brown-York energy in the hyperbolic setting (proved by Shi and Tam) will be important.
\begin{theorem}[{\cite[Theorem 3.8]{shi2007rigidity}}]\label{rigidityBYhyper}
Let \( (\Omega, g) \) be a compact manifold with smooth boundary \( \Sigma \). Assume the following conditions hold:
\begin{itemize}
    \item[(i)] The scalar curvature \( \mathrm{Sc}^\Omega \) of \( (\Omega, g) \) satisfies \( \mathrm{Sc}^M \geq 2\Lambda  \) for some \( \Lambda <  0 \).
    \item[(ii)] \( \Sigma \) is a topological sphere with Gaussian curvature \( K > \frac{\Lambda}{3} \) and with positive mean curvature \( H \).
\end{itemize}
 Then there exists an isometric embedding of $\Sigma$ into the hyperbolic space of radius $\tfrac{\Lambda}{3}$, $\mathbb{H}^3_{\Lambda/3}$, with image a convex surface of mean curvature $H_0$. Furthermore,
\begin{equation}
\int_{\Sigma} (H_0 - H) \, d\mu \geq 0
\end{equation}
 and equality holds if and only if \( (\Sigma, g) \) is a domain of    $\mathbb{H}^3_{\Lambda/3}$.
\end{theorem}
With this, we can prove. 
\begin{theorem}\label{rigiditywillhyper}
     Let $(M,g)$ be a  $3$-dimensional Riemannian  manifold with scalar curvature $\mathrm{Sc}^M \geq 2 \Lambda  $ for a constant $\Lambda < 0$. Suppose $\Omega$ is a relatively compact domain with smooth connected boundary  $\Sigma=  \partial \Omega $.  Let $\Sigma $ be an area-constrained Willmore surface with positive mean curvature and  Lagrange parameter $\lambda  \geq -\frac{2}{3}\Lambda$, that is, it  satisfies
\begin{equation}\label{willmoreequa2}
        0= \lambda H  +\Delta^\Sigma H + H|\mathring{B}|^2+ H \mathrm{Ric}^M(\nu, \nu), \quad \lambda  \geq -\frac{2}{3}\Lambda, \quad H>0.
    \end{equation}
   Then $\mathcal{E}_\Lambda(\Sigma)\ge0$, and if $\mathcal{E}_\Lambda(\Sigma)=0$,  then  $\Omega$ is isometric to a geodesic ball in the hyperbolic space of radius $3/\Lambda$, $\mathbb{H}^3_{\Lambda/3}$.
\end{theorem}
\begin{proof}
We proceed as in the proof of Theorem \ref{rigiditywillmore}.
We  multiply equation (\ref{willmoreequa2}) by \( H^{-1} \), integrate  by parts, and use Gauss equation getting. 
\[
\lambda |\Sigma | + \int_{\Sigma} |\nabla^{\Sigma}  \log H|^2 + \frac{1}{4} H^2 + \frac{1}{2} |\mathring{B}|^2  \, d\mu \leq 4\pi - \int_{\Sigma } \frac{1}{2} \, \mathrm{Sc}^M \, d\mu.
\]   

First we want to see that $\int_{\Sigma } H^2+\frac{4}{3} \Lambda d\mu  \leq 16 \pi   $, adding and subtracting $\Lambda$ we obtain  
\[
(\lambda +\Lambda )|\Sigma |   + \int_{\Sigma} |\nabla^{\Sigma}  \log H|^2 + \frac{1}{4} H^2 + \frac{1}{2} |\mathring{B}|^2  \, d\mu \leq 4\pi - \int_{\Sigma } \frac{1}{2} \, (\mathrm{Sc}^M -2 \Lambda )\, d\mu.
\]   
and using that $ \lambda  \geq -\frac{2}{3}\Lambda $ and  $\mathrm{Sc}^M \geq 2 \Lambda  $ we have
\[
 \frac{1}{4}\int_{\Sigma } H^2+\frac{1}{3} \Lambda d\mu \leq (\lambda +\Lambda )|\Sigma |   + \int_{\Sigma} |\nabla^{\Sigma}  \log H|^2 + \frac{1}{4} H^2 + \frac{1}{2} |\mathring{B}|^2  \, d\mu \leq 4\pi .
\]   
and with this, we have the nonnegativity. 

Now, if $\mathcal{E}_\Lambda(\Sigma)=0$, then   $\int_{\Sigma } H^2+\frac{4}{3} \Lambda d\mu =16 \pi   $ and using $\mathrm{Sc}^M \geq 2 \Lambda  $ we obtain 
$$\lambda |\Sigma | + \int_{\Sigma} |\nabla^{\Sigma}  \log H|^2  + \frac{1}{2} |\mathring{B}|^2  \, d\mu \leq -\frac{2}{3}\Lambda |\Sigma| $$
 Now as $\lambda  \geq -\frac{2}{3}\Lambda$ 
 this implies that $\lambda  = -\frac{2}{3}\Lambda$, $|\mathring{B}|  = 0 $, $\mathrm{Sc}^M_{| \Sigma} = 2 \Lambda$ and $H$ is constant. Then by (\ref{willmoreequa2}) we also have $ \mathrm{Ric}^M(\nu, \nu)=\frac{2}{3}\Lambda$ along $\Sigma $. By the vanishing Hawking energy we have $H^2 = -\frac{4}{3}\Lambda  + 16 \pi |\Sigma|^{-1} $, in particular, by the Gauss equation we have $ \mathrm{Sc}^{\Sigma} = \frac{2}{r^2}$ where $r$ is the area radius of $\Sigma$.  Now with this, we can apply the rigidity result of Theorem \ref{rigidityBYhyper} to get the result.    
\end{proof}

For the case $\Lambda > 0$, the "zero-energy" model space to compare for rigidity  will be the standard round sphere  $\mathbb{S}^n(r)$ of radius $r$ and we will denote by $\mathbb{S}_+^n(r):=\{x\in \mathbb{R}^{n+1} : \, |x|=r, \, x_{n+1} \geq 0  \}$ its upper hemisphere. Note that the Hawking energy $ \mathcal{E}_{\Lambda}$ with cosmological constant $\Lambda > 0$ is the "most negative" Hawking energy from the ones we have considered so far; therefore, we will need stronger assumptions to obtain rigidity results.  First, we introduce the following results by Hang and Wang, which lies at the core of the proofs of rigidity for $\Lambda >0$. 
\begin{theorem}[Hang-Wang {\cite[Theorem 3]{hang2009rigidity}}]\label{hangwang3}
Let $(M,g)$ be a smooth compact $n$-dimensional Riemannian manifold with boundary $\partial M = \Sigma$, and let $\Omega \subset \mathbb{S}^n_+$ be a compact domain with smooth boundary in the open hemisphere. Suppose:
\begin{itemize}
    \item $\mathrm{Ric}^M \ge (n-1)g$,
    \item there is an isometric embedding $\iota \colon (\Sigma, g_\Sigma) \to \partial\Omega$,
    \item $B \ge \iota^*B_0$, where $B$ is the second fundamental form of $\Sigma$ in $M$ and $B_0$ is the second fundamental form of $\partial\Omega$ in $\mathbb{S}^n_+(1)$.
\end{itemize}
Then $(M,g)$ is isometric to $(\Omega, g_{\mathbb{S}^n_+})$.
\end{theorem}
In this setting,  Melo proved the following rigidity result for stable constant mean curvature surfaces.
\begin{theorem}[{\cite[Theorem 1.1, Theorem 1.2]{melo2024hawking}}] Let $(M,g)$ be a complete $3$-dimensional Riemannian manifold with scalar curvature $\mathrm{Sc}^M \ge 6$, and let $\Sigma\subset M$ be a stable constant mean curvature sphere satisfying $ \int_\Sigma (H^2+4)\,d\mu = 16\pi$.  Assume that either \begin{enumerate} \item[(i)] $\Sigma$ admits an even symmetry, or \item[(ii)] the Gauss curvature $K_\Sigma$ is sufficiently $\mathcal C^0$-close to $\frac{4\pi}{|\Sigma|}$. \end{enumerate} Then $\Sigma$ is isometric to the round sphere of radius  \( \sqrt{|\Sigma|/ 4\pi} \). Moreover, if $\mathrm{Ric}^M\ge 2g$, then the mean-convex region $\Omega\subset M$ bounded by $\Sigma$ is isometric to a geodesic ball in $\mathbb S^3_+$.
\end{theorem}
In the case of Willmore surfaces we have the following result.
\begin{theorem}\label{positivecosmo}
Let $(M,g)$ be a $3$-dimensional Riemannian manifold, and let $\Lambda > 0$ be a constant. Let $\Sigma \subset M$ be a closed connected area-constrained Willmore surface with non-negative mean curvature ($H \ge 0$) that is not identically zero, and with Lagrange parameter $\lambda \ge -\frac{2}{3}\Lambda$. 
\begin{itemize}
     \item[(i)] If $H>0$ and $\mathrm{Sc}^M\ge 2\Lambda$ along $\Sigma$, then \begin{equation}\label{spheinequality}
        \int_\Sigma H^2 d\mu \le 16\pi - \frac{4}{3}\Lambda|\Sigma|.
    \end{equation}
    Equivalently, the Hawking energy with cosmological constant is nonnegative. 
     
     \item[(ii)] If $ \mathrm{Ric}^M(\nu,\nu)\ge \frac23\Lambda$ along $\Sigma$, where $\nu$ is the outward unit normal, then $\Sigma$ is totally umbilic, $H$ is a strictly positive constant, $ \mathrm{Ric}^M(\nu,\nu)=\frac23\Lambda$ on $\Sigma$,  and $\lambda=-\frac23\Lambda$.
    
    \item[(iii)] Suppose additionally  that $\Sigma=\partial\Omega$ is the boundary of a relatively compact domain $\Omega\subset M$. If $\mathrm{Ric}^M \ge \frac{2}{3}\Lambda g$ on $\Omega$ and equality holds in \eqref{spheinequality}, then $\Sigma$ is isometric to a round sphere and $\Omega$ is globally isometric to a  geodesic ball in the round sphere $\mathbb{S}^3(R)$, where $R = \sqrt{3/\Lambda}$.
\end{itemize}
\end{theorem}
\begin{proof}
$(i)$ Note that the proof of the inequality $\int_\Sigma H^2 d\mu \le 16\pi - \frac{4}{3}\Lambda|\Sigma|$ in Theorem \ref{rigiditywillhyper} does not depend on the sign of $\Lambda$; therefore, under these assumptions the inequality also holds.

$(ii)$ Since $\Sigma$ is not minimal,  there is a point $p \in \Sigma$ such that $H(p)\neq 0$. Integrating the Willmore equation and using the bound on the Ricci tensor and the Lagrange parameter, we have
$$0\geq \int_\Sigma H\Big(\lambda + \frac{2}{3}\Lambda +|\mathring{B}|^2 \Big) d\mu\geq 0.$$
As the integrand is nonnegative, we have $H(\lambda + \frac{2}{3}\Lambda +|\mathring{B}|^2 )=0$. Since $H(p) > 0$, evaluating at $p$ gives $\lambda = - \frac{2}{3}\Lambda$ and $|\mathring{B}|^2(p)=0$. This implies $H |\mathring{B}|^2=0$ identically. Putting this information back into the Willmore equation and integrating, we have 
$$0= \int_\Sigma H\Big( \mathrm{Ric}^M(\nu, \nu) - \frac{2}{3}\Lambda \Big) d\mu.$$
Again, as the integrand is nonnegative, $H( \mathrm{Ric}^M(\nu, \nu) - \frac{2}{3}\Lambda )=0$. The Willmore equation then reduces to $\Delta^\Sigma H =0$. Since $\Sigma$ is compact without boundary, the maximum principle dictates that $H$ is a positive constant. Because $H$ is a strictly positive constant, it immediately follows that $\mathrm{Ric}^M(\nu, \nu) = \frac{2}{3}\Lambda$ and $|\mathring{B}|^2=0$ everywhere on $\Sigma$. 

 $(iii)$ Substituting these vanishing terms into the Gauss equation gives:
\begin{equation}\label{Gaussequalam}
    \mathrm{Sc}^{\Sigma} = \mathrm{Sc}^{M} - 2\mathrm{Ric}^M(\nu, \nu) + \frac{1}{2}H^2 - |\mathring{B}|^2 = \mathrm{Sc}^{M} - \frac{4}{3}\Lambda + \frac{1}{2}H^2.
\end{equation}
Assuming equality holds in \eqref{spheinequality}, the vanishing Hawking energy condition $\mathcal{E}_\Lambda(\Sigma) = 0$ is equivalent to $\frac{1}{2}H^2 |\Sigma| = 8\pi - \frac{2}{3}\Lambda|\Sigma|$. Integrating \eqref{Gaussequalam} over $\Sigma$ and applying the Gauss-Bonnet theorem yields:
$$8\pi  \geq \int_\Sigma \mathrm{Sc}^{\Sigma} d\mu = \int_\Sigma \Big(\mathrm{Sc}^{M} -  \frac{4}{3}\Lambda\Big) d\mu + \frac{1}{2}H^2|\Sigma| = 8\pi + \int_\Sigma (\mathrm{Sc}^{M} - 2\Lambda) \, d\mu.$$
Since $\mathrm{Sc}^{M} \geq 2\Lambda$, the integral on the right is non-negative, implying that $8\pi  =\int_\Sigma \mathrm{Sc}^{\Sigma} d\mu $ (so $\Sigma$ is a topological sphere) and $\mathrm{Sc}^{M} = 2\Lambda$ on $\Sigma$. Returning to the Gauss equation, $\mathrm{Sc}^{\Sigma} = \frac{2}{3}\Lambda + \frac{1}{2}H^2$ is a positive constant. By the Weyl-Nirenberg-Pogorelov theorem \ref{WeyNi}, $\Sigma$ embeds isometrically into Euclidean space as a round sphere of area radius $r = \sqrt{|\Sigma|/(4\pi)}$. Thus, $\mathrm{Sc}^\Sigma = \frac{2}{r^2}$.

Let $R = \sqrt{3/\Lambda}$ be the radius of the reference space form $\mathbb{S}^3(R)$. Substituting $\frac{2}{3}\Lambda = \frac{2}{R^2}$ into our deduced Gauss equation gives $\frac{2}{r^2} = \frac{2}{R^2} + \frac{1}{2}H^2$, which yields $H^2 = 4\left(\frac{1}{r^2} - \frac{1}{R^2}\right),$ which agrees with the mean curvature of a geodesic sphere of area radius $r$ in a sphere of radius $R$.
Since $H > 0$, this implies $r < R$. Because $\mathring{B} = 0$, $\Sigma$ is totally umbilic with a non-negative second fundamental form ($B \ge 0$). 

Rescaling the manifold by $\tilde{g} := R^{-2}g$ normalizes the ambient Ricci bound to $\mathrm{Ric}^{\tilde{g}} \ge 2\tilde{g}$. The boundary $\Sigma$ becomes a round sphere of intrinsic radius $\tilde{r} = r/R < 1$. Let $\Omega_0 \subset \mathbb{S}^3_+$ be a  geodesic ball in the open unit hemisphere whose boundary has intrinsic radius $\tilde{r}$. The squared mean curvature of $\partial\Omega_0$ is $\tilde{H}_0^2 = 4(1/\tilde{r}^2 - 1)$. Under our rescaling, $\tilde{H}^2 = R^2 H^2 = 4 (1/\tilde{r}^2 - 1) = \tilde{H}_0^2$. Since $H > 0$, the signs match, yielding $\tilde{H} = \tilde{H}_0$. All hypotheses of Theorem \ref{hangwang3} are satisfied, meaning $(\Omega, \tilde{g})$ is globally isometric to $\Omega_0 \subset \mathbb{S}^3_+$. Undoing the scaling, $(\Omega, g)$ is isometric to a  geodesic ball in $\mathbb{S}^3(R)$.
\end{proof}
The preceding theorem treats the nonminimal case $H>0$. The limiting case $H\equiv 0$ leads instead to full hemisphere rigidity. We postpone this endpoint case to the higher-dimensional section (see Theorem \ref{thm:minimal_hemisphere_rigidity_higher}), where the argument is naturally formulated for general hypersurfaces, and where we prove that the vanishing of the corresponding spherical Hawking energy forces the enclosed region to be a hemisphere.

Analogous to the previous section, by applying Definition \ref{electrichar} to this setting, we can also consider the general charged cosmological constant case. In this setting, the dominant energy condition reduces to $\mathrm{Sc}^M \ge 2 \Lambda +2|E|^2$. We then have directly the following unified positivity and rigidity result.

\begin{corollary}
    Let $(M,g,E)$ be a $3$-dimensional time-symmetric initial data for the Einstein-Maxwell equations with cosmological constant $\Lambda \in \mathbb{R}$ satisfying the dominant energy condition $\mathrm{Sc}^M \ge 2 \Lambda +2|E|^2$. Let $\Omega \subset M$ be a relatively compact domain whose smooth connected boundary $\Sigma = \partial \Omega$ is an area-constrained Willmore surface with positive mean curvature and Lagrange parameter $\lambda \ge -\frac{2}{3}\Lambda$. 
    
    Then the charged cosmological Hawking energy 
    \begin{equation}
        \mathcal{E}_{Q, \Lambda}(\Sigma) = \sqrt{\frac{|\Sigma|}{16\pi}} 
        \left( 
        1 + \frac{4\pi Q(\Sigma)^2}{|\Sigma|}  
        - \frac{1}{16\pi} \int_{\Sigma} \left( H^2 + \frac{4}{3} \Lambda \right) \, d\mu 
        \right)
    \end{equation}
    is nonnegative. Furthermore, if $\mathcal{E}_{Q, \Lambda}(\Sigma) = 0$, then $E=0$ on $\Omega$ and:
    \begin{itemize}
        \item[i)] If $\Lambda \le 0$, $\Omega$ is isometric to a geodesic ball in the hyperbolic space $\mathbb{H}^3_{\Lambda/3}$ (or Euclidean space if $\Lambda=0$).
        \item[ii)] If $\Lambda>0$ and, in addition, $\mathrm{Ric}^M\ge \frac{2}{3}\Lambda g$ on $\Omega$, $\Omega$ is isometric to a geodesic ball in the round sphere $\mathbb{S}^3(R)$ with $R = \sqrt{3/\Lambda}$.
    \end{itemize}
\end{corollary}
\begin{proof}
    Since the dominant energy condition ensures $\mathrm{Sc}^M \ge 2 \Lambda + 2|E|^2 \ge 2 \Lambda$, we immediately have $\mathcal{E}_{Q, \Lambda}(\Sigma) \ge \mathcal{E}_{\Lambda}(\Sigma) \ge 0$. The vanishing of $\mathcal{E}_{Q, \Lambda}(\Sigma)$ therefore implies the vanishing of the uncharged Hawking energy, $\mathcal{E}_{\Lambda}(\Sigma) = 0$. Applying the preceding rigidity results for the uncharged cases, $\Omega$ is isometric to the corresponding model region. Because these model spaces have constant scalar curvature $\mathrm{Sc}^M = 2\Lambda$, the dominant energy condition $\mathrm{Sc}^M \ge 2 \Lambda + 2|E|^2$ forces $2|E|^2 \le 0$, and hence $E=0$ on $\Omega$.
\end{proof}

\begin{remark}
    As mentioned in Remark \ref{chargeremark}, one could also consider a magnetic field $\mathcal{B}$, and the result would follow in a similar manner. Note that this particular variation of the Hawking energy was also considered in \cite{baltazar2023local, leandro2025sharp, sousa2023charged}, although the rigidity results presented there are quite different.
\end{remark}

\subsection{Higher dimensional case}\label{higherdimsec}

First, we need to note that in higher dimensions the equation characterizing  Willmore surfaces changes. When considering an $n-1$-dimensional hypersurface $\Sigma$ in an $n$-dimensional Riemannian manifold $(M,g)$, $\Sigma$ is an area-constrained Willmore hypersurface if it satisfies  
\begin{equation}\label{willmorehigherdim}
    0= \lambda H  +\Delta^\Sigma H- \frac{n-3}{2(n-1)}H^3 + H|\mathring{B}|^2+ H \mathrm{Ric}^M(\nu, \nu) 
\end{equation}
this comes directly by considering a variation of the Willmore functional, consider a general variation  \( \frac{\partial F}{\partial s} \Big|_{s=0} = \alpha \nu \) and $\int_{\Sigma} \alpha H \, d\mu = 0$ then 
\begin{equation}
   \frac{1}{2} \frac{\partial }{\partial s}  \int_{\Sigma_s}H^2d\mu = \int_{\Sigma_s}( - \Delta^\Sigma H- \frac{H^3}{(n-1)} - H|\mathring{B}|^2- H \mathrm{Ric}^M(\nu, \nu) +\frac{1}{2} H^3) \alpha \,
  d\mu 
\end{equation}
and we get the equation directly. The equation  is like the $2$-dimensional Willmore equation  but with the extra term $ \frac{n-3}{2(n-1)}H^3 $. Note also that because of this extra term, a round sphere in $\mathbb{R}^n$ is a Willmore surface with Lagrange parameter $\lambda= \frac{(n-3)(n-1)}{2r^2}$ where $r$ is the area radius of $\Sigma$.  

Motivated by the higher-dimensional Hawking energies considered in \cite{miao2017quasi} and \cite{pacheco2023constructing}, we define two natural generalized Hawking energies for an $(n-1)$-dimensional hypersurface subject to a cosmological constant $\Lambda$: 
\begin{equation}\label{energy_1_general}
\mathcal{E}_{n,\Lambda,1}(\Sigma) = \frac{1}{2(n-1)(n-2) \omega_{n-1}}\left( \frac{|\Sigma|}{\omega_{n-1}} \right)^{\frac{1}{n-1}}\int_{\Sigma}\left( \mathrm{Sc}^{\Sigma} - \frac{n-2}{n-1} H^2 - \frac{2(n-2)}{n}\Lambda \right)d\mu\end{equation}
and
\begin{equation}\label{energy_2_general}
\mathcal{E}_{n,\Lambda,2}(\Sigma) = \frac{1}{2} \left( \frac{|\Sigma|}{\omega_{n-1}} \right)^{\frac{n-2}{n-1}} \left(1  - \frac{1}{(n-1)^2 \omega_{n-1}} \left( \frac{\omega_{n-1}}{ |\Sigma|} \right)^{\frac{n-3}{n-1}} \int_{\Sigma } \left(H^2+\frac{2(n-1)}{n} \Lambda\right)d\mu \right).
\end{equation}
where $\omega_{n-1}$ is the volume of the $n-1$-dimensional round sphere. 
Note that both of them reduce to the Hawking energy in dimension $n=3$ in the case of a sphere and also satisfy several key features of the Hawking energy.  Note also that the second one can be seen as a natural generalization when thinking on the  Willmore functional.

First, we establish the nonnegativity of these energies under appropriate bounds on the Lagrange parameter.
\begin{theorem}\label{thm:unified_positivity}
Let $(M,g)$ be an $n$-dimensional Riemannian manifold ($n\geq 3$), and let $\Lambda \in \mathbb{R}$ be a constant. Let $\Sigma \subset M$ be a closed, area-constrained Willmore hypersurface with strictly positive mean curvature ($H > 0$) and parameter $\lambda$. Assume that the ambient scalar curvature satisfies $\mathrm{Sc}^M \ge 2\Lambda$ along $\Sigma$.
\begin{itemize}
    \item[$(i)$] If $\label{unified_bound_1}
    \lambda \geq \frac{n-3}{2(n-2)|\Sigma|} \int_{\Sigma} \mathrm{Sc}^{\Sigma} d\mu - \frac{n-1}{n}\Lambda$,
then $\mathcal{E}_{n,\Lambda,1}(\Sigma)\geq 0$.
    \item[$(ii)$] If $
    \lambda \geq \frac{1}{2|\Sigma|}\int_{\Sigma} \mathrm{Sc}^{\Sigma} d\mu - \frac{n-1}{2} \left( \frac{\omega_{n-1}}{ |\Sigma|} \right)^{\frac{2}{n-1}} - \frac{n-1}{n}\Lambda$,
then $\mathcal{E}_{n,\Lambda,2}(\Sigma)\geq 0$.
\end{itemize}
Furthermore, if either inequality is strict, the respective Hawking energy is strictly positive.
\end{theorem}
\begin{proof}
Because $H > 0$ strictly, we may multiply the Willmore equation \eqref{willmorehigherdim} by $H^{-1}$. Integrating the Laplacian term by parts and using the Gauss equation to substitute $\mathrm{Ric}^M(\nu, \nu)$ yields the fundamental integral identity:
\begin{equation}\label{willmoredived_unified}
\lambda |\Sigma | + \int_{\Sigma} |\nabla^{\Sigma} \log H|^2 + \frac{1}{2(n-1)} H^2 + \frac{1}{2} |\mathring{B}|^2 \, d\mu = \frac{1}{2} \int_{\Sigma } (\mathrm{Sc}^{\Sigma} - \mathrm{Sc}^M) \, d\mu.
\end{equation}
Applying the ambient scalar curvature bound $\mathrm{Sc}^M \ge 2\Lambda$, the right-hand side is strictly bounded above by $\frac{1}{2} \int_{\Sigma } \mathrm{Sc}^{\Sigma} d\mu - \Lambda|\Sigma|$. 

Assuming the first condition  for $\lambda$, we obtain
\begin{equation*}
\begin{split}
    \int_{\Sigma} \frac{1}{2(n-1)} H^2 \, d\mu &\leq \frac{1}{2} \int_{\Sigma} \mathrm{Sc}^{\Sigma} d\mu - \Lambda|\Sigma| - \lambda|\Sigma| \\
    &\leq \frac{n-2 - (n-3)}{2(n-2)} \int_{\Sigma} \mathrm{Sc}^{\Sigma} d\mu - \left(1 - \frac{n-1}{n}\right)\Lambda|\Sigma| \\
    &= \frac{1}{2(n-2)} \int_{\Sigma} \mathrm{Sc}^{\Sigma} d\mu - \frac{1}{n}\Lambda|\Sigma|.
\end{split}
\end{equation*}
Multiplying this inequality algebraically by $2(n-1)(n-2)/|\Sigma|$  implies $\mathcal{E}_{n,\Lambda,1}(\Sigma)\geq 0$.

If instead we assume the second condition  for $\lambda$, the same substitution yields
\begin{equation*}
\begin{split}
    \int_{\Sigma} \frac{1}{2(n-1)} H^2 \, d\mu &\leq \frac{1}{2} \int_{\Sigma} \mathrm{Sc}^{\Sigma} d\mu - \Lambda|\Sigma| - \lambda|\Sigma| \\
    &\leq \frac{n-1}{2} \left( \frac{\omega_{n-1}}{ |\Sigma|} \right)^{\frac{2}{n-1}}|\Sigma| - \frac{1}{n}\Lambda|\Sigma|.
\end{split}
\end{equation*}
Dividing this by $\frac{n-1}{2}$  recovers the bound $\mathcal{E}_{n,\Lambda,2}(\Sigma)\geq 0$, completing the proof.
\end{proof}
We now specialize the preceding result to the flat reference case \(\Lambda=0\). For simplicity, we write
$
\mathcal E_{n,i}:=\mathcal E_{n,0,i}$,  $i=1,2$. In this case the corresponding rigidity statement is the following.
\begin{theorem}\label{rigidityhighdim}
    Let $(M,g)$ be an $n$-dimensional Riemannian  manifold (with $n\geq 3$) with $\mathrm{Sc}^M \ge 0$, and let  $\Omega \subset M$ be a relatively compact domain whose smooth boundary  $\partial \Omega $  has finitely many components. Suppose each boundary component has positive mean curvature and admits an isometric embedding into \(\mathbb{R}^n\) as a convex hypersurface.  Suppose further that one of these components $\Sigma $ is an area-constrained Willmore surface with parameter $\lambda$ and that either  
    \begin{itemize}
    \item[$i)$]$\mathcal{E}_{n,1}(\Sigma)=0$ and $
        \lambda  \geq  \frac{n-3}{2(n-2)|\Sigma|}  
  \int_{\Sigma} 
 \mathrm{Sc}^{\Sigma} d\mu$,
     or
    \item[$ii)$] $\mathcal{E}_{n,2}(\Sigma)= 0$ and $
     \lambda  \geq  \frac{1}{2|\Sigma|}\int_{\Sigma} 
 \mathrm{Sc}^{\Sigma} d\mu  -\frac{n-1}{2} \left( \frac{\omega_{n-1}}{ |\Sigma|} \right)^{\frac{2}{n-1}}$.
\end{itemize}
Then $\partial \Omega$ is connected and isometric to a round sphere, and  $\Omega$ is isometric to a ball in $\mathbb{R}^{n}$.
\end{theorem}
\begin{proof}
Since $\Lambda=0$, Theorem \ref{thm:unified_positivity} applies under either of the two stated assumptions. Thus $\mathcal E_{n,i}(\Sigma)\ge 0$. If $\mathcal E_{n,i}(\Sigma)=0$, then the chain of inequalities in the proof of Theorem \ref{thm:unified_positivity} is saturated. In particular, the corresponding lower bound for $\lambda$ is attained, and the nonnegative terms in \eqref{willmoredived_unified} vanish. Hence $\nabla^\Sigma\log H=0$, $\mathring B=0$, and $\mathrm{Sc}^M=0$ along $\Sigma$. Therefore $H$ is constant and $\Sigma$ is totally umbilic.

We now show that $\Sigma$ has constant scalar curvature. Since $H$ is constant and $\mathring B=0$, the Willmore equation \eqref{willmorehigherdim}, divided by $H$, reduces to
\begin{equation}
0
=
\lambda
-
\frac{n-3}{2(n-1)}H^2
+
\mathrm{Ric}^M(\nu,\nu).
\end{equation}
Using also $\mathrm{Sc}^M=0$ and $\mathring B=0$, the Gauss equation gives
\begin{equation}
\mathrm{Sc}^{\Sigma}
=
-2\mathrm{Ric}^M(\nu,\nu)
+
\frac{n-2}{n-1}H^2.
\end{equation}
Combining the last two identities yields
$
\mathrm{Sc}^{\Sigma}
=
2\lambda
+
\frac{1}{n-1}H^2$. 
Thus $\mathrm{Sc}^{\Sigma}$ is constant. Since $\Sigma$ admits an isometric embedding into $\mathbb R^n$ as a convex hypersurface, Ros's theorem \cite{ros1988compact} implies that the Euclidean embedding is a round sphere.

Let $H_0$ denote the mean curvature of this Euclidean embedding, and let $r=\left(|\Sigma|/\omega_{n-1}\right)^{\frac{1}{n-1}}$ be its area radius. Since the Euclidean embedding is round, $H_0=\frac{n-1}{r}$ and
$\mathrm{Sc}^{\Sigma}=\frac{n-2}{n-1}H_0^2$.

It remains to show that $H=H_0$. If $\mathcal E_{n,1}(\Sigma)=0$, then, since $H$ and $\mathrm{Sc}^{\Sigma}$ are constant, the definition of $\mathcal E_{n,1}$ gives $
\mathrm{Sc}^{\Sigma}
=
\frac{n-2}{n-1}H^2$. 
Comparing with $\mathrm{Sc}^{\Sigma}=\frac{n-2}{n-1}H_0^2$  gives $H^2=H_0^2$.  If instead $\mathcal E_{n,2}(\Sigma)=0$, then
$
H^2
=
(n-1)^2
\left(
\frac{\omega_{n-1}}{|\Sigma|}
\right)^{\frac{2}{n-1}}
=
H_0^2$. In both cases, since $H>0$ and $H_0>0$, we obtain $H=H_0$. The rigidity statement in Theorem \ref{highershitam} now applies. Hence $\partial\Omega$ is connected, $\Sigma$ is a round sphere, and $\Omega$ is isometric to a Euclidean ball.
\end{proof}
Note that this result depends on Theorem \ref{highershitam}, which, in turn, relies on the positive mass theorem in higher dimensions.

So far we have two different conditions for $\mathcal{E}_{n,i}$, but we could also find a common condition for $\lambda$ and $\mathrm{Sc}^{\Sigma} $ so that the previous two theorems hold.
\begin{corollary}
Let $(M,g)$ be a $n$-dimensional Riemannian  manifold (with $n\geq 3$) with nonnegative scalar curvature. Let $\Sigma$ be an area-constrained Willmore surface satisfying 
    \begin{equation}
        H>0, \quad \lambda\geq  \frac{(n-3)(n-1)}{2} \left( \frac{\omega_{n-1}}{ |\Sigma|} \right)^{\frac{2}{n-1}}  \quad \text{and} \quad   \frac{1}{|\Sigma| } \int_{\Sigma} 
\mathrm{Sc}^{\Sigma} d\mu\leq (n-1)(n-2) \left( \frac{\omega_{n-1}}{ |\Sigma|} \right)^{\frac{2}{n-1}}. 
    \end{equation}
Then $\Sigma$    satisfies $ \mathcal{E}_{n,i}(\Sigma) \geq 0$ for $i=1,2$.  If additionally  $\Sigma $ is the boundary of a relatively compact domain that can be isometrically embedded in \(\mathbb{R}^n\) as a convex hypersurface. And  either $\mathcal{E}_{n,i}(\Sigma) = 0 $ for $i=1$ or $i=2$ then   $\Omega$ is isometric to a ball in $\mathbb{R}^{n}$ and $\Sigma$ is isometric to a round sphere.
\end{corollary}
Similar to Corollary \ref{positivemasshaw}, one can produce a similar positive mass theorem for Willmore surfaces in higher dimensions, We can also consider the case with charge. In higher dimensions, the dominant energy condition for a charged manifold is given by $\mathrm{Sc}^M \geq (n-1)(n-2)|E|^2 $, and we can generalize the previous Hawking energies to   
\begin{equation}
\begin{split}
 \mathcal{E}_{n,Q,1}(\Sigma)=   \frac{1}{2(n-1)(n-2) \omega_{n-1}} 
\left( \frac{|\Sigma|}{\omega_{n-1}} \right)^{\frac{1}{n-1}}
\int_{\Sigma} 
\Big(& \mathrm{Sc}^{\Sigma} + (n-1)(n-2) \left( \frac{\omega_{n-1}}{ |\Sigma|} \right)^2 Q(\Sigma)^2\\
& \,- \frac{n-2}{n-1} H^2 \Big)
d\mu
\end{split}
\end{equation}
and 
\begin{equation}
\mathcal{E}_{n,Q,2}(\Sigma) =    \frac{1}{2} \left( \frac{|\Sigma|}{\omega_{n-1}} \right)^{\frac{n-2}{n-1}} \left(1 + Q(\Sigma)^2 \left( \frac{\omega_{n-1}}{ |\Sigma|} \right)^{\frac{2(n-2)}{n-1}} - \frac{1}{(n-1)^2 \omega_{n-1}} \left( \frac{\omega_{n-1}}{ |\Sigma|} \right)^{\frac{n-3}{n-1}} \int_{\Sigma }H^2  d\mu \right),
\end{equation}
 which was already derived in \cite{pacheco2023constructing}. Then we have, as a direct consequence of Theorem \ref{highershitam} the following:
 \begin{corollary}
    Let $(M,g,E)$ be a $n$-dimensional (with $n\geq 3$) time-symmetric initial data for the Einstein-Maxwell equations  which satisfies the dominant energy condition $ \mathrm{Sc}^M \geq (n-1)(n-2) |E|^2$,  and let  $\Omega \subset M$ be a relatively compact domain whose smooth boundary  $\partial \Omega $  has finitely many components. Suppose each boundary component has positive mean curvature and admits an isometric embedding into \(\mathbb{R}^n\) as a convex hypersurface. Suppose further that one of these components $\Sigma$ is an area-constrained Willmore and that either
    \begin{itemize}
    \item[$i)$]  $\mathcal{E}_{n,Q,1}(\Sigma)=0$ and  $
        \lambda  \geq  \frac{n-3}{2(n-2)|\Sigma|}  
  \int_{\Sigma} 
 \mathrm{Sc}^{\Sigma} d\mu$, or
    \item[$ii)$]  $\mathcal{E}_{n,Q,2}(\Sigma)= 0$ and   $
     \lambda  \geq  \frac{1}{2|\Sigma|}\int_{\Sigma} 
 \mathrm{Sc}^{\Sigma} d\mu  -\frac{n-1}{2} \left( \frac{\omega_{n-1}}{ |\Sigma|} \right)^{\frac{2}{n-1}}$.
\end{itemize}
Then  $\partial \Omega$ is connected and isometric to a round sphere, $\Omega$ is isometric to a Euclidean ball in $\mathbb{R}^n$, and $ E$ vanishes  on $\Omega$.  
 \end{corollary}

 Finally, we address the rigidity of these energies in the spherical setting, where the cosmological constant is strictly positive ($\Lambda > 0$). In this regime, the global rigidity of the domain is governed by the Hang-Wang theorems rather than the Positive Mass Theorem, necessitating the stronger ambient Ricci lower bound.
\begin{theorem}\label{thm:rigidityhighsphdim}
Let $(M,g)$ be an $n$-dimensional Riemannian manifold ($n\geq 3$), and let $\Lambda > 0$ be a constant. Let $\Omega \subset M$ be a relatively compact domain satisfying $\mathrm{Ric}^M \ge \frac{2}{n}\Lambda g$, whose boundary $\Sigma = \partial\Omega$ is a connected, area-constrained Willmore hypersurface with strictly positive mean curvature ($H > 0$) and parameter $\lambda$. Suppose $\Sigma$ admits an isometric embedding into $\mathbb{R}^n$ as a convex hypersurface. If either:
\begin{itemize}
    \item[$i)$] $\mathcal{E}_{n,\Lambda,1}(\Sigma)=0$ and $ \lambda \geq \frac{n-3}{2(n-2)|\Sigma|} \int_{\Sigma} \mathrm{Sc}^{\Sigma} d\mu - \frac{n-1}{n}\Lambda$, or
    \item[$ii)$] $\mathcal{E}_{n,\Lambda,2}(\Sigma)= 0$ and $
        \lambda \geq \frac{1}{2|\Sigma|}\int_{\Sigma} \mathrm{Sc}^{\Sigma} d\mu - \frac{n-1}{2} \left( \frac{\omega_{n-1}}{ |\Sigma|} \right)^{\frac{2}{n-1}} - \frac{n-1}{n}\Lambda$.
\end{itemize}
Then $\Sigma$ is intrinsically a round sphere, and $\Omega$ is globally isometric to a  geodesic ball in the round space form $\mathbb{S}^n(R)$, where $R = \sqrt{\frac{n(n-1)}{2\Lambda}}$.
\end{theorem}
\begin{proof}
Note that  $\mathrm{Ric}^M \ge \frac{2}{n}\Lambda g$ implies  $\mathrm{Sc}^M \ge 2\Lambda$  on $\Omega$. If either Hawking energy vanishes ($\mathcal{E}_{n,\Lambda,i}(\Sigma) = 0$), the entire inequality chain from the proof of Theorem \ref{thm:unified_positivity} must hold as  equalities. Because we assumed the respective lower bounds on $\lambda$, the chain of inequalities  is saturated, in particular, the lower bound for $\lambda$ is attained.  Simultaneously, this saturation forces all non-negative geometric terms in the fundamental identity \eqref{willmoredived_unified} to vanish: $|\mathring{B}|^2 = 0$, $\nabla^\Sigma \log H = 0$ (so $H$ is a strictly positive constant), and $\mathrm{Sc}^M = 2\Lambda$ everywhere on $\Sigma$. 

Because the ambient Ricci tensor satisfies $\mathrm{Ric}^M \ge \frac{2}{n}\Lambda g$ and its trace evaluates to exactly $n(\frac{2}{n}\Lambda) = 2\Lambda$ on the boundary, the Ricci tensor is  pure trace along $\Sigma$. Thus, $\mathrm{Ric}^M(\nu,\nu) = \frac{2}{n}\Lambda$ along $\Sigma$.

Substituting these values into the Gauss equation gives 
\begin{equation*}
    \mathrm{Sc}^{\Sigma} = \mathrm{Sc}^{M} - 2\mathrm{Ric}^M(\nu, \nu) + \frac{n-2}{n-1}H^2 - |\mathring{B}|^2 = 2\Lambda - \frac{4}{n}\Lambda + \frac{n-2}{n-1}H^2.
\end{equation*}
Since $H$ is constant, $\mathrm{Sc}^\Sigma$ is a strictly positive constant. Because $\Sigma$ is a closed hypersurface admitting a convex isometric embedding into $\mathbb{R}^n$ with constant scalar curvature, Ros's Constant-Scalar-Curvature Rigidity Theorem \cite{ros1988compact} dictates that $\Sigma$ is intrinsically a round sphere. Consequently, its area radius $r = \left(|\Sigma|/ \omega_{n-1} \right)^{\frac{1}{n-1}}$ is uniquely determined and constant, yielding $\mathrm{Sc}^\Sigma = \frac{(n-1)(n-2)}{r^2}$.

Let $R = \sqrt{\frac{n(n-1)}{2\Lambda}}$ be the radius of the reference space form $\mathbb{S}^n(R)$. Since $H$ is constant and $\Sigma$ is intrinsically round, substituting the scalar curvature directly into either vanishing energy condition ($\mathcal{E}_{n,\Lambda,i}(\Sigma) = 0$) algebraically reduces exactly to 
$H^2 = (n-1)^2\left(\frac{1}{r^2} - \frac{1}{R^2}\right),$ which agrees with the mean curvature of a geodesic sphere of  area radius $r$ in a sphere of radius $R$. Since $H > 0$, this strict positivity implies $r < R$. Furthermore, because $\mathring{B} = 0$, the boundary is totally umbilic with a non-negative second fundamental form ($B \ge 0$).

We now apply the conformal rescaling $\tilde{g} := R^{-2}g$. This normalizes the ambient Ricci lower bound to $\mathrm{Ric}^{\tilde{g}} \ge (n-1)\tilde{g}$ and maps $\Sigma$ to a round sphere of intrinsic radius $\tilde{r} = r/R < 1$. Because $H$  agrees with the mean curvature of the boundary of a geodesic ball $\Omega_0 \subset \mathbb{S}^n_+$ whose boundary has area radius $\tilde{r}$, we may apply the Hang-Wang rigidity Theorem \ref{hangwang3}. This dictates that $(\Omega, \tilde{g})$ is globally isometric to the  geodesic ball $\Omega_0$. Undoing the conformal scaling demonstrates that $(\Omega, g)$ is globally isometric to a geodesic ball in $\mathbb{S}^n(R)$, completing the proof.
\end{proof}

In the degenerate case where the mean curvature vanishes identically ($H \equiv 0$),  the area-constrained Willmore equation becomes identically satisfied and therefore does not by itself impose the umbilicity or roundness needed for rigidity.  To handle this limiting geometry, we rely on the complementary rigidity theorem of Hang and Wang for the exact hemispherical boundary.
\begin{theorem}[Hang-Wang {\cite[Theorem 2]{hang2009rigidity}}]\label{hangwang}
Let $(M,g)$ be a compact $n$-dimensional Riemannian manifold ($n\ge2$) with nonempty boundary $\Sigma$. Assume $\mathrm{Ric}^M \ge (n-1)g$, $(\Sigma,g_\Sigma)$ is isometric to the standard unit sphere $\mathbb{S}^{n-1}(1)$, and the second fundamental form of $\Sigma$ is nonnegative. Then $(M,g)$ is isometric to the hemisphere $\mathbb{S}^n_+(1)$.
\end{theorem}
This result is a Ricci-strengthened version of Min-Oo’s conjecture \cite[Theorem 4]{min1998scalar}. Note that having the full conjecture would allow us to have stronger rigidity results in the case $\Lambda >0$, however, the original conjecture, phrased purely in terms of a scalar curvature lower bound, was later disproved  by Brendle, Marques and Neves \cite{brendle2011deformations}.

Global rigidity to the full hemisphere can be recovered for minimal hypersurfaces by analyzing the scalar curvature constraints inherent in the generalized Hawking energies.
\begin{theorem}\label{thm:minimal_hemisphere_rigidity_higher}
Let $(M,g)$ be an $n$-dimensional Riemannian manifold ($n \ge 3$), and let $\Lambda > 0$ be a constant. Let $\Omega \subset M$ be a relatively compact domain satisfying $\mathrm{Ric}^M \ge \frac{2}{n}\Lambda g$, whose  smooth connected boundary $\Sigma = \partial\Omega$ is a minimal hypersurface ($H \equiv 0$) admitting an isometric embedding into $\mathbb{R}^n$ as a convex hypersurface. 

Suppose that the ambient scalar curvature satisfies $\mathrm{Sc}^M = 2\Lambda$ along $\Sigma$. If either $\mathcal{E}_{n,\Lambda,1}(\Sigma) = 0$ or $\mathcal{E}_{n,\Lambda,2}(\Sigma) = 0$, then $\Sigma$ is totally geodesic and intrinsically round, and $(\Omega,g)$ is globally isometric to the hemisphere $\mathbb{S}^n_+(R)$, where $R = \sqrt{\frac{n(n-1)}{2\Lambda}}$.
\end{theorem}
\begin{proof}
Applying the boundary hypothesis $\mathrm{Sc}^M|_\Sigma = 2\Lambda$, the ambient Ricci lower bound and the Gauss equation we find 
\begin{equation}\label{gauss_upper_bound_minimal}
    \mathrm{Sc}^\Sigma \le 2\Lambda - \frac{4}{n}\Lambda - |\mathring{B}|^2 = \frac{2(n-2)}{n}\Lambda - |\mathring{B}|^2.
\end{equation}
Integrating this inequality over $\Sigma$ establishes an  upper bound on the total intrinsic scalar curvature:
\begin{equation}\label{integrated_gauss_upper}
    \int_\Sigma \mathrm{Sc}^\Sigma d\mu \le \frac{2(n-2)}{n}\Lambda|\Sigma| - \int_\Sigma |\mathring{B}|^2 d\mu.
\end{equation}
We now examine the two vanishing energy conditions. 
If $\mathcal{E}_{n,\Lambda,1}(\Sigma) = 0$, the definition of the first Hawking energy immediately yields the exact scalar-curvature integral equality:
\begin{equation}\label{energy_1_equality}
    \int_\Sigma \mathrm{Sc}^\Sigma d\mu = \frac{2(n-2)}{n}\Lambda|\Sigma|.
\end{equation}
If instead $\mathcal{E}_{n,\Lambda,2}(\Sigma) = 0$, we obtain 
$|\Sigma| = \omega_{n-1} \left( \frac{n(n-1)}{2\Lambda} \right)^{\frac{n-1}{2}}$. By the classical Alexandrov-Fenchel inequalities for quermassintegrals of convex domains in Euclidean space, applied to the convex Euclidean embedding of $\Sigma$ (see, for instance, \cite[Equation (4)]{guan2009quermassintegral}), one obtains
\begin{equation}\label{alexandrov_fenchel}
\int_\Sigma \mathrm{Sc}^{\Sigma}\,d\mu
\ge
(n-1)(n-2)\omega_{n-1}
\left(
\frac{|\Sigma|}{\omega_{n-1}}
\right)^{\frac{n-3}{n-1}}.
\end{equation}  Substituting the fixed area  into the right-hand side of \eqref{alexandrov_fenchel} evaluates to  $\frac{2(n-2)}{n}\Lambda|\Sigma|$. Thus, the vanishing of the second Hawking energy provides the lower bound
\begin{equation}\label{energy_2_lower_bound}
    \int_\Sigma \mathrm{Sc}^\Sigma d\mu \ge \frac{2(n-2)}{n}\Lambda|\Sigma|.
\end{equation}
In either case \eqref{energy_1_equality} or \eqref{energy_2_lower_bound}, comparison with the integrated Gauss upper bound \eqref{integrated_gauss_upper} forces equality throughout. More precisely, since $\mathrm{Sc}^M=2\Lambda$ along $\Sigma$ and $H=0$, the Gauss equation gives
\begin{equation}
\mathrm{Sc}^{\Sigma}
=
\frac{2(n-2)}{n}\Lambda
-
|\mathring B|^2
-
2\left(\mathrm{Ric}^M(\nu,\nu)-\frac{2}{n}\Lambda\right).
\end{equation}
Since $\mathrm{Ric}^M\ge \frac{2}{n}\Lambda g$, equality in the integrated inequality implies $\mathring B=0$ and $\mathrm{Ric}^M(\nu,\nu)=\frac{2}{n}\Lambda$ along $\Sigma$. Hence $\mathrm{Sc}^{\Sigma}=\frac{2(n-2)}{n}\Lambda$ pointwise. Since $H=0$ and $\mathring B=0$, we have $B=0$, so $\Sigma$ is totally geodesic.

Moreover, $\Sigma$ has constant positive scalar curvature and admits a convex isometric embedding into $\mathbb R^n$. By Ros's theorem \cite{ros1988compact}, $\Sigma$ is intrinsically round. If $r=\left(|\Sigma|/\omega_{n-1}\right)^{\frac{1}{n-1}}$ is its area radius, then $\mathrm{Sc}^{\Sigma}=\frac{(n-1)(n-2)}{r^2}$. Comparing this with $\mathrm{Sc}^{\Sigma}=\frac{2(n-2)}{n}\Lambda$ gives $r^2=\frac{n(n-1)}{2\Lambda}$. Thus, for $R=\sqrt{\frac{n(n-1)}{2\Lambda}}$, one has $r=R$ and $|\Sigma|=\omega_{n-1}R^{n-1}$.

We now apply the conformal rescaling $\tilde{g} := R^{-2}g$, which normalizes the ambient Ricci lower bound to $\mathrm{Ric}^{\tilde{g}} \ge (n-1)\tilde{g}$. Under this metric scaling, the $(n-1)$-dimensional boundary area scales by $R^{-(n-1)}$, yielding $
    |\Sigma|_{\tilde{g}} = R^{-(n-1)} |\Sigma| = R^{-(n-1)} \left( \omega_{n-1} R^{n-1} \right) = \omega_{n-1}$.
    
Thus, the connected boundary is intrinsically isometric to the standard unit sphere $\mathbb{S}^{n-1}(1)$ and possesses a vanishing second fundamental form ($\tilde{B} = 0$). Applying the Hang-Wang rigidity Theorem \ref{hangwang}, we conclude that $(\Omega, \tilde{g})$ is globally isometric to the full unit hemisphere $\mathbb{S}^n_+$. Undoing the scaling demonstrates that $(\Omega, g)$ is globally isometric to the hemisphere $\mathbb{S}^n_+(R)$, completing the proof.
\end{proof}
\begin{remark}
In dimension $3$, the convex embedding assumption in Theorem \ref{thm:minimal_hemisphere_rigidity_higher} can be replaced by the weaker assumption that $\Sigma$ has spherical topology. When $H\equiv 0$ and $\mathcal E_\Lambda(\Sigma)=0$, one obtains $|\Sigma|=12\pi/\Lambda$. The Gauss equation, together with $\mathrm{Sc}^M=2\Lambda$ along $\Sigma$ and $\mathrm{Ric}^M\ge \frac{2}{3}\Lambda g$, gives $\mathrm{Sc}^{\Sigma}\le \frac{2}{3}\Lambda$. Since $\Sigma$ is a topological sphere, Gauss-Bonnet gives $\int_\Sigma \mathrm{Sc}^{\Sigma}\,d\mu=8\pi=\frac{2}{3}\Lambda|\Sigma|$, and hence equality holds pointwise. Thus $\Sigma$ is intrinsically round, and the same Hang-Wang argument gives hemisphere rigidity.
\end{remark}
\begin{remark}
    The charged positive cosmological constant case follows directly from the preceding results. If the charged dominant energy condition $\mathrm{Sc}^M\ge 2\Lambda+(n-1)(n-2)|E|^2$
holds, the vanishing of the corresponding charged Hawking energy forces the vanishing of the uncharged one, and the preceding rigidity theorem applies. In particular, the charged dominant energy condition then forces $E=0$ on $\Omega$.
\end{remark}
In the time-symmetric case, a comparison between (almost round) stable CMC surfaces and Willmore surfaces reveals that both satisfy positivity and rigidity results, among other key properties. However, when extending to the general dynamical setting, it becomes unclear how to generalize stable CMC surfaces, as the stability condition for CMC surfaces does not have a straightforward analog for STCMC or constant expansion surfaces.

\section{Dynamical setting (\texorpdfstring{$k\neq 0$} )) }\label{sectiondyn}

 The dynamical or nontotally geodesic  setting is more challenging since the tensor $k$ is something in principle external representing the extrinsic geometry of $(M,g)$ when embedded in a spacetime, and the only way to connect it to the intrinsic geometry of $(M,g)$ is using the Einstein constrained equations and some condition in the matter content like the dominant energy condition.

 First, we derive the equation that characterizes the area surface equations of the Hawking functional in dimension $3$. This was already done in \cite[Lemma 2.1]{diaz2023local}, but we include it for completeness.
\begin{lemma}[First variation]
The area-constrained Euler Lagrange equation for the Hawking functional (\ref{hawfun}) is 
\begin{equation}\label{eulag}
\begin{split}
   0=& \lambda H  +\Delta^\Sigma H + H|\mathring{B}|^2+ H \mathrm{Ric}^M(\nu, \nu)+P( \nabla_\nu \tr k - \nabla_\nu k(\nu,\nu )) - 2P \diver_\Sigma (k(\cdot, \nu))\\ & +\frac{1}{2}H P^2 - 2k (\nabla^\Sigma P, \nu )
    \end{split}
\end{equation}
Here $H$ is the mean curvature of $\Sigma$ , $\mathring{B}$   is the traceless part of the second
fundamental form $B$ of $\Sigma$ in $M$, that is
$\mathring{B}= B- \frac{1}{2} H g_\Sigma$ where $g_\Sigma$ is the induced
metric on $\Sigma$,  $\mathrm{Ric}^M $ is the Ricci curvature of $M$, $\nabla^\Sigma $, $\diver_\Sigma$ and $\Delta^\Sigma$ are the covariant derivative, tangential divergence  and Laplace Beltrami operator on $\Sigma$. Finally $\lambda \in \mathbb{R}$ plays the role of a Lagrange parameter.
\end{lemma}
\begin{proof}
Let $\Sigma \subset M$ be a surface and let $f: \Sigma \times (-\epsilon , \epsilon) \rightarrow M$ be a variation of $\Sigma$ with $f(\Sigma, s)= \Sigma_s$ and lapse $\frac{\partial f}{\partial s }_{|s=0}=\alpha \nu $. In \cite[Section 3]{willflat}, it was shown that the first variation of the Willmore functional (\ref{willfunc}) is given by 
\begin{equation}
\begin{split}
  \frac{1}{2}  \frac{d}{d s}\int_{\Sigma_s} H^2  d\mu_{| s=0}  = \int_{\Sigma_s} \left( -\Delta^\Sigma  H - H|\mathring{B}|^2- H\mathrm{Ric}^M(\nu, \nu)  \right) \alpha \, d\mu,
    \end{split}
\end{equation}
now let's compute the variation of $\frac{1}{2} \int_\Sigma P^2  d\mu $. In \cite{Ce}, it was shown that the variation of $P$ is given by
\begin{equation}
    \frac{d\, P }{d s}_{| s=0}=\left( \nabla_\nu \tr k - \nabla_\nu k(\nu, \nu)\right)\alpha +2 k(\nabla \alpha, \nu),
\end{equation}
using this relation and integration by parts we have 
\begin{equation}
\begin{split}
 \frac{1}{2}  \frac{d}{d s}\int_{\Sigma_s} P^2  d\mu_{| s=0}  =& \int_{\Sigma_s} \frac{1}{2} P^2 H  \alpha +P\left( \nabla_\nu \tr k - \nabla_\nu k(\nu, \nu)\right)\alpha +2P k(\nabla \alpha, \nu)  d\mu\\
 =\int_{\Sigma_s}\big( &\frac{1}{2} P^2 H   +P\left( \nabla_\nu \tr k - \nabla_\nu k(\nu, \nu)\right) -2P \diver_\Sigma \left( k(\cdot, \nu) \right )\\
 &- 2 k(\nabla^\Sigma P, \nu ) \big) \alpha  d\mu.
    \end{split}
\end{equation}
We are considering area-constrained surfaces, which means surfaces whose variation of area is zero. This traduces to the area-constrained $\int_\Sigma H \alpha d\mu =0 $. Then our surfaces must satisfy the area-constrained and 
\begin{equation*}
\begin{split}
 &0=\frac{1}{2}  \left( \frac{d}{d s}\int_{\Sigma_s} H^2  d\mu_{|s=0} -  \frac{d}{d s}\int_{\Sigma_s} P^2  d\mu_{|s=0} \right)  =\\
 &\int_{\Sigma_s} \big( -\Delta^\Sigma H - H|\mathring{B}|^2- H \mathrm{Ric}^M(\nu, \nu) - \frac{1}{2} P^2 H   -P\left( \nabla_\nu \tr k - \nabla_\nu k(\nu, \nu)\right) +2P \diver_\Sigma \left( k(\cdot, \nu) \right )\\
 &+ 2 k(\nabla^\Sigma P, \nu ) \big) \alpha  d\mu
    \end{split}
\end{equation*}
 Then combining  this expression and the area-constrained gives us the Euler Lagrange equation (\ref{eulag}).
\end{proof}
Finally, note that this is equivalent to \cite[Lemma 2.8]{Alex}, and  that in the time‐symmetric case  it reduces to the Willmore equation (\ref{Willeq}). For the most general spacetime variation of the Hawking energy (including $\Lambda\neq0$ and for spacetime flows of any causal character) see \cite{bray2007generalized}.

Although S. Hawking himself did not work specifically with area-constrained critical surfaces of the generalized Willmore functional $ \int_\Sigma H^2 -P^2 d\mu$, we will refer to these surfaces as Hawking surfaces. This terminology is chosen because, as we will see, their defining properties align naturally with the Hawking energy, making them particularly well-suited for its analysis.
\begin{definition}
 We call the surfaces satisfying  equation   (\ref{eulag}) \emph{Hawking surfaces}.  
\end{definition}
    Hawking surfaces are defined within a given spacelike hypersurface of spacetime. This implies that if we wish to define them independently of a specific hypersurface i.e., as purely spacelike $2$-surfaces in spacetime, we must select a preferred spacelike normal direction to perform the variation. Consequently, this introduces a degree of gauge dependence into the definition.

Now we will study the positivity of the Hawking energy under these surfaces. First, recall that the dominant energy condition is given by 
\begin{equation}
    \mu \geq |J|
\end{equation}
where 
\begin{equation}
     \mathrm{Sc}^M + (\tr k)^2 - |k|^2=2 \mu \quad \text{and} \quad  \diver (k -(\tr k) g)= J
\end{equation}
are the energy density and the momentum density of the Einstein constraint equations.

The search for a physically meaningful quasi-local energy in general relativity has led to numerous proposals. One of the most natural approaches is to follow the Hamilton–Jacobi method, which was first used by Brown and York in \cite{brown1993quasilocal} to derive a quasi-local energy expression. However, the Hamilton–Jacobi method alone does not yield a unique quasi-local energy formulation, it requires additional choices, such as a reference configuration and a generator vector field for the physical quantity being measured.

An alternative perspective was introduced by Kijowski in \cite{kijowski1997simple}, who proposed a different reference configuration and vector field, leading to a new quasi-local energy formulation. Later, Liu and Yau in \cite{liu2003positivity} refined Kijowski’s definition, demonstrating that it satisfies key physical requirements, such as positivity and well-posedness under general conditions.

Similar to the Brown-York energy, the Kijowski-Liu-Yau energy relies on an isometric embedding theorem: a closed spacelike $2$-surface 
$\Sigma$ is embedded into Euclidean $3$-space, and its extrinsic curvature is compared with that of the physical spacetime. However, unlike the Brown-York energy, the Kijowski-Liu-Yau energy has the advantage of being gauge invariant.

In an initial data set setting, we have: 
 \begin{definition}
     Consider a surface $\Sigma$ with positive Gauss curvature, which is contained in an initial data set $(M,g, k)$ and satisfies $H^2 -P^2 \geq 0$, then its 
 \emph{Kijowski-Liu-Yau energy} is given by 
     $$\mathcal{E}_{KLY}(\Sigma)=\frac{1}{8 \pi G}\int_\Sigma  H_0- \sqrt{H^2-P^2} d\mu $$
     where $H_0$ is the mean curvature of the surface when isometrically embedded into $\mathbb{R}^3$ and $G$ is the gravitational  constant.
 \end{definition}
 We will use the rigidity of the Kijowski-Liu-Yau energy, which was proven by Liu and Yau and can be written in our notation as 
\begin{theorem}[{\cite[Theorem 1]{liu2003positivity, liu2006positivity}}] \label{liuyaurigi}
    Let \( (\Omega, g, k)\) be a $3$-dimensional compact initial data set satisfying the dominant energy condition, such that its boundary \(\partial \Omega\) has finitely many connected components \(\Sigma_1, \dots, \Sigma_\ell\), each of which has positive Gaussian curvature and a spacelike mean curvature vector ($H^2-P^2>0$).  Then \(\mathcal{E}_{KLY}(\Sigma_\alpha) \geq 0\) for \(\alpha = 1, \dots, \ell\). Moreover, if \(\mathcal{E}_{KLY}(\Sigma_\alpha) = 0\) for some \(\alpha\), then  \(\partial \Omega\) is connected and  $\Omega$ is isometric to a spacelike hypersurface in Minkowski
spacetime with second fundamental form  $k$.  
\end{theorem}
This is a remarkable result; however, the Kijowski-Liu-Yau energy has the drawback of being too positive, meaning that one can find surfaces in Minkowski spacetime where the Kijowski-Liu-Yau energy is strictly positive. This issue was first demonstrated by Ó Murchadha and Szabados in \cite{murchadha2004comment} and was later fully characterized by Miao, Shi, and Tam in the following result.
\begin{theorem}[{\cite[Theorem 4.1]{miao2010geometric}}]\label{minkowskirigi}
    Let \( \Sigma \) be a closed, connected, smooth, spacelike \( 2 \)-surface in Minkowski spacetime \( \mathbb{R}^{3,1} \). 
    Suppose \( \Sigma \) spans a compact spacelike hypersurface in \( \mathbb{R}^{3,1} \).
    If \( \Sigma \) has positive Gaussian curvature and a spacelike mean curvature vector ($H^2-P^2>0$), then $
    \mathcal{E}_{KLY} (\Sigma) \geq 0$.
    Moreover, \( \mathcal{E}_{KLY} (\Sigma) = 0 \) if and only if \( \Sigma \) lies on a hyperplane in \( \mathbb{R}^{3,1} \).
\end{theorem}
Now we are going to derive nonnegativity and rigidity results for the Hawking energy on Hawking surfaces. In this case we will require an extra technical condition on a new quantity $f$.
\begin{theorem}\label{positivity0}
    Let $(M,g,k)$ be a $3$-dimensional initial data set satisfying the dominant energy condition.
    
    $i)$ Let $\Sigma$ be a Hawking surface with positive mean curvature, and such that for 
   $$ f:= \left( \frac{P}{H}\right)^2|k|^2+ \frac{1}{2 }(\tr k)^2    - \frac{3}{4} P^2- \frac{P}{H}( \nabla_\nu \tr k - \nabla_\nu k(\nu,\nu )) -  \frac{1}{2} |k|^2  -\frac{1}{2} |\mathring{B}|^2 -|J|$$
    the surface satisfies $\int_\Sigma f -\lambda d\mu \leq 0$. Then $\int_\Sigma H^2 -P^2 d\mu \leq 16\pi $, and  if $\int_\Sigma f -\lambda d\mu < 0$ then  $\int_\Sigma H^2 -P^2 d\mu < 16\pi $. In particular, the Hawking energy is nonnegative.

 $ii)$ Let  $\Omega \subset M$ be a relatively compact domain whose smooth boundary  $\partial \Omega $  has finitely many components. Suppose that one of the boundary components $\Sigma$ is a Hawking surface with positive mean curvature, and the other components have positive scalar curvature and spacelike mean curvature vector ($H^2-P^2>0$). If   there exists a constant $0\leq\beta <\frac{1}{2}$ such that $\int_\Sigma f_\beta -\lambda \, d\mu \leq 0$  for 
       $$ f_\beta:= \left( \frac{P}{H}\right)^2|k|^2+ \frac{1}{2 }(\tr k)^2   - \frac{3}{4} P^2- \frac{P}{H}( \nabla_\nu \tr k - \nabla_\nu k(\nu,\nu ))  -  \beta( |k|^2 + |\mathring{B}|^2 +2|J|), $$
and  $\int_\Sigma H^2 -P^2 d\mu = 16\pi $. Then   $\Omega$ is isometric to a spacelike hypersurface in Minkowski
spacetime with second fundamental form  $k$,   $\partial \Omega$ is connected ($ \partial \Omega=\Sigma $), umbilic,  isometric to a round sphere and $k=0$ on $\Sigma$.
\end{theorem}
\begin{proof}
  $(i)$ We proceed similarly as in the previous proofs of this section. We consider  equation (\ref{eulag}), divide it by $H$, integrate by parts the term $\frac{\Delta^\Sigma H}{H}$ and use the Gauss equation $2\mathrm{Ric}^M(\nu, \nu) = \mathrm{Sc} -\mathrm{Sc}^{\Sigma}  + \frac{1}{2}H^2 - |\mathring{B}|^2    $ obtaining
\begin{equation}\label{argudyna}
\begin{split}
   0=  \int_{\Sigma}&\lambda+ |\nabla^\Sigma \log H|^2 + \frac{1}{2} |\mathring{B}|^2+  \frac{1}{2}(\mathrm{Sc}^M -\mathrm{Sc}^{\Sigma}) +\frac{P}{H}( \nabla_\nu \tr k - \nabla_\nu k(\nu,\nu ))\\& +\frac{1}{4}H^2 + \frac{1}{2} P^2 - 2\frac{P}{H} \diver_\Sigma (k(\cdot, \nu))- \frac{2}{H} k (\nabla^\Sigma P, \nu )d\mu. 
    \end{split}
\end{equation}
Now using  Gauss-Bonnet theorem to replace $ \mathrm{Sc}^{\Sigma}$, adding and subtracting $(\tr k)^2$, $|k|^2$ and $|J|$ and integrating by parts  we have 
\begin{equation*}\label{posirigiequ}
\begin{split}
    \frac{1}{4}\int_\Sigma H^2-P^2 d\mu \leq& 4\pi +\int_\Sigma -\frac{1}{2}(\mathrm{Sc}^M + (\tr k)^2 - |k|^2-2|J|)-\frac{P}{H}( \nabla_\nu \tr k - \nabla_\nu k(\nu,\nu )) -\lambda \\& +\frac{1}{2 }(\tr k)^2 - \frac{1}{2}|k|^2 -|J|  -|\nabla^\Sigma  \log H|^2 - \frac{1}{2} |\mathring{B}|^2  - \frac{3}{4} P^2 + \frac{2P}{H}  k(\nabla^\Sigma \log H, \nu)  d\mu\\
    =&4\pi +\int_\Sigma -(\mu-|J|) +f-\lambda+g\, d\mu
    \end{split}
\end{equation*}
where 
$$g:= -\left( \frac{P}{H}\right)^2|k|^2 -|\nabla^\Sigma  \log H|^2 + \frac{2P}{H}  k(\nabla^\Sigma \log H, \nu). $$
Then we need to see that the integral is nonpositive, by assumption the first two terms are nonpositive and it is direct to see that $g\leq 0$, then we have the first result.

$(ii)$ As before  we can write 
\begin{equation}
    \frac{1}{4}\int_\Sigma H^2-P^2 d\mu \leq 4\pi +\int_\Sigma -(\mu-|J|) +f_\beta-\lambda+g -(\frac{1}{2}-\beta) (|k|^2  + |\mathring{B}|^2+\frac{1}{2} |J|)  \,   d\mu.
\end{equation}
Then if $\int_\Sigma H^2-P^2 d\mu = 16\pi$, as $\frac{1}{2}-\beta> 0 $ and using  our assumptions  we obtain  $|k|^2=|\mathring{B}|^2=|J|=0 $ on $\Sigma$, this also implies that  $\lambda= \mathrm{Sc}^M_{| \Sigma} =0$,  that $H=\frac{2}{r}$ is constant, where $r$ is the area radius of $\Sigma$, and that $\Sigma$ is a sphere since $\int_{\Sigma } \frac{1}{2}\mathrm{Sc}^{\Sigma} \, d\mu = 4 \pi$.  Now equation (\ref{eulag}) forces $\mathrm{Ric}^M(\nu, \nu) =0$ on $\Sigma$. Now by Gauss equation, we have that the isometric embedding of $\Sigma$ into Euclidean spaces has the scalar curvature of a round sphere and therefore it is a round sphere. Then $H=H_0$ and $P=0$, and by the rigidity of Theorem \ref{liuyaurigi} the result follows.
\end{proof}
\begin{remark}\label{remarf}
    Note that one could define $f$ differently, 
\begin{equation}\label{ftil}
    \Tilde{f}:= \frac{2P}{H}  k(\nabla^\Sigma \log H, \nu)+ \frac{1}{2 }(\tr k)^2    - \frac{3}{4} P^2- \frac{P}{H}( \nabla_\nu \tr k - \nabla_\nu k(\nu,\nu )) -  \frac{1}{2} |k|^2  -\frac{1}{2} |\mathring{B}|^2 -|J|.
\end{equation}
In this case, the function $g$ of the proof would be $g=-|\nabla^\Sigma  \log H|^2 $ and 
 by requiring $\int_\Sigma \Tilde{f} -\lambda \,d\mu \leq 0$, one also obtains nonnegativity of the Hawking energy. The same argument applies if one replaces $ \Tilde{f}$ by an  analogous   $ \Tilde{f}_{\beta}$, yielding an identical rigidity conclusion. Although $\Tilde{f}$ isolates better the terms governing the sign of the Hawking mass (giving a more precise condition), it involves more cumbersome surface‐gradient calculations. We therefore employ the simpler function 
$f$, which delivers the same positivity and rigidity conclusions with far less technical overhead.
\end{remark}
\begin{remark}
    Note that the condition \( \int_\Sigma f_\beta -\lambda \, d\mu \leq 0 \) is a strengthening of \( \int_\Sigma f -\lambda \, d\mu \leq 0 \). Neither of these conditions is optimal nor physically motivated. In particular, the function $f_\beta$ was introduced to ensure that, in the case where the Hawking energy vanishes, it follows that  $k_{|\Sigma}=0$. This allows us to apply the Willmore equation, from which we can conclude that the surface $\Sigma$ has positive Gaussian curvature. Consequently, the rigidity result of the Kijowski-Liu-Yau energy becomes applicable. However, as we will see in Remark \ref{tooposi}, this condition might be artificially enforcing the Hawking energy to be too positive.
\end{remark}
Similar to the time-symmetric case and Corollary \ref{positivemasshaw} we can formulate a positive energy theorem for the Hawking energy on Hawking surfaces for the dynamical setting. 
\begin{corollary}\label{positivemasshawdyn}
  Let $(M,g,k)$ be a $3$-dimensional initial data set satisfying the dominant energy condition. Suppose $\Omega$ is a relatively compact domain with smooth connected boundary  $\Sigma=  \partial \Omega $. If $\Sigma$ is a  Hawking surface with Lagrange parameter $\lambda$ and there exists a constant $\beta <\frac{1}{2}$ such that $\int_\Sigma f_\beta -\lambda \, d\mu \leq 0$, then
  $$  \mathcal{E}(\Sigma) = \sqrt{\frac{|\Sigma|}{16 \pi}} \left( 1- \frac{1}{16 \pi} \int_\Sigma H^2-P^2  d\mu \right) \geq 0 $$
  with equality if and only if $\Omega$ is isometric to a spacelike hypersurface  in Minkowski spacetime with second fundamental form $k$, $\Sigma$ is an umbilic round sphere, and $k=0$ along $\Sigma$.  
\end{corollary}
\begin{proof}
The first part is a direct consequence of Theorem \ref{positivity0}, note that in the proof one also obtains  $|\mathring{B}|=0$ in the case of equality, so $\Sigma$ is umbilic. What remains to show is that if $\Sigma$ is am umbilic round sphere with $k=0$ along $\Sigma$ and $\Omega$ is isometric to a spacelike hypersurface  in Minkowski spacetime, then $\mathcal{E}(\Sigma) =0 $.

First, note that since $\Omega$ is  in Minkowski space and $k=0$ along $\Sigma$ by Gauss-Codazzi equation  $\mathrm{Ric}^\Omega=0 $. Then by Gauss equation $\frac{2}{r^2}=\mathrm{Sc}^\Sigma = \frac{1}{2}H^2- |\mathring{B}|^2 =\frac{1}{2}H^2 $, where $r$ is the area radius of $\Sigma$. Then a direct computation  gives $\mathcal{E}(\Sigma) =0 $. 
\end{proof}
We now examine the behavior of the Hawking energy when evaluated on the surfaces considered in the rigidity result of Theorem \ref{positivity0}. We will observe that it tends to be excessively positive, meaning that there exist numerous surfaces in Minkowski space with strictly positive Hawking energy. This phenomenon mirrors the well-known over-positivity issue of the Kijowski-Liu-Yau energy. Given that our argument relies on its rigidity result, it is unsurprising that by combining Theorems \ref{minkowskirigi} and \ref{positivity0}, we arrive at the following result. 
\begin{corollary}
  Let $(M,g,k)$ be a $3$-dimensional compact hypersurface in Minkowski spacetime. Assume that the boundary of $M $, $\partial M =\Sigma $   is a Hawking surface of positive mean curvature and that there exists a constant $\beta <\frac{1}{2}$ such that $\int_\Sigma f_\beta -\lambda \, d\mu \leq 0$. Then the Hawking energy on $\Sigma$ is strictly positive unless $\Sigma$ is also  contained in a hyperplane of Minkowski spacetime.  
\end{corollary}
\begin{proof}
    Suppose that the Hawking energy of $\Sigma$ vanishes and it is not contained in a hyperplane of Minkowski spacetime. In the proof of Theorem \ref{positivity0}, we saw that under the condition  $\int_\Sigma f_\beta -\lambda \, d\mu \leq 0$, the vanishing of the Hawking energy implies the vanishing of the Kijowski-Liu-Yau energy, then by Theorem \ref{minkowskirigi} we get a contradiction.
\end{proof}

\begin{remark}\label{tooposi} Paradoxically, the primary motivation for considering Hawking surfaces was to address the issue of the Hawking energy being too negative. However, we now find that, under certain conditions, the Hawking energy can become excessively positive. There are two possible explanations for this phenomenon:

\begin{enumerate}
    \item \textbf{Issues with the condition} \( \int_\Sigma f_\beta -\lambda \, d\mu \leq 0 \):  
    This was introduced as a technical refinement of the weaker condition \( \int_\Sigma f -\lambda \, d\mu \leq 0 \). The formulation with the parameter \( \beta \) was specifically chosen to ensure that \( k \) vanishes on \( \Sigma \), allowing for a clearer geometric characterization of the surface. However, this is by no means an optimal or physically motivated condition. This condition may impose an overly restrictive constraint, biasing the selection of Hawking surfaces toward those with higher energy. A better choice of condition could potentially lead to a stronger rigidity result—one that does not rely on the rigidity properties of the Kijowski-Liu-Yau energy.

    \item \textbf{Potential excess in the Hawking energy measurement:}  
    Alternatively, it is possible that the Hawking energy on these surfaces is genuinely "too positive," meaning that Hawking surfaces may introduce an excess in its measurement.    
\end{enumerate}
However, as we will see in the following examples in Minkowski spacetime, it looks like the issue lies in a too restrictive condition.
\end{remark}
\begin{example}[Hyperboloid]\label{example1}
    In $4$-dimensional Minkowski spacetime  \( \mathbb{R}^{3,1} \), we consider for some positive  constant $a$ the hyperboloid.
$$M=\{(t,x,y,z) \in \mathbb{R}^{3,1} : t^2-x^2-y^2-z^2 =a^2,\, t>0\} $$
The induced metric is the metric of the hyperbolic space $\mathbb{H}^3_{a}$ which in polar coordinates is given by 
\begin{equation}
    g^M = \frac{dr^2}{1 + \frac{r^2}{a^2}}
+
r^2 \bigl(d\theta^2 + \sin^2\!\theta\,d\phi^2\bigr).
\end{equation}
where $r^2=x^2+y^2+z^2$.  The second fundamental form of $M$ in $\mathbb{R}^{3,1} $ is given by 
\begin{equation}
    k=\frac{1}{a} g^M
\end{equation}
    Then $M$ is totally umbilic and  $\tr k = \frac{3}{a}$. We will consider spheres of constant radius $\Sigma_r= \{x^2+y^2+z^2=r^2 \}$, these surfaces are round spheres of area $|\Sigma|=4\pi r^2$. We will see that they are Hawking surfaces with vanishing Hawking energy. Since the normal is given by $\nu = \sqrt{1 +\frac{r^2}{a^2}} \partial_r$, using the spherical symmetric we can see that 
\begin{equation}
    H(\Sigma_r)= \frac{2}{r} \sqrt{1 +\frac{r^2}{a^2}},  \quad P(\Sigma_r)= \frac{2}{a}. \quad \text{and} \quad  |\mathring{B}|=0
\end{equation}
Then $\Sigma_r$ are in particular constant mean curvature (CMC) surfaces and also spacetime constant mean curvature (STCMC) surfaces, also it is direct to see that they have vanishing Hawking energy ($\mathcal{E}(\Sigma_r)=0$).  Now to see that they are Hawking surface, note that $\mathrm{Ric}^M= -\frac{2}{a} g^M$, then  $\mathrm{Ric}^M(\nu, \nu)=-\frac{2}{a} $ and since $k$ is constant along $\Sigma$ the equation characterizing the Hawking surfaces reduces to 
\begin{equation}
 0=  \lambda H+ H\mathrm{Ric}^M(\nu, \nu) + \frac{1}{2}P^2H =\lambda H
\end{equation}
which holds for $\lambda=0$. Then we have that $\Sigma_r$ is a Hawking surface with vanishing Hawking energy (note also that $\Sigma_r$ is a Willmore surface for $\lambda= \frac{2}{a}$). Finally, we calculate the function $f$, using that $|k|^2= \frac{3}{a^2}$ we can see 
\begin{equation}
    f= \left( \frac{P}{H}\right)^2|k|^2+ \frac{1}{2 }(\tr k)^2  - \frac{1}{2}|k|^2  - \frac{3}{4} P^2= \left( \frac{P}{H}\right)^2|k|^2 >0. 
\end{equation}
Then we have $\int_{\Sigma_r} f-\lambda d\mu= \int_{\Sigma_r} \left( \frac{P}{H}\right)^2|k|^2 d\mu >0 $. Note  that this shows in particular that the hyperboloid is foliated by Hawking spheres of zero Hawking energy, and since these surfaces are also STCMC surfaces, it is also foliated by these surfaces. 
\end{example}
\begin{example}\label{example2}
      In $4$-dimensional Minkowski spacetime  \( \mathbb{R}^{3,1} \), we consider for some constant $\alpha>0$ the hypersurface
$$M=\{(t,x,y,z) \in \mathbb{R}^{3,1} : t=\frac{\alpha}{2} (x^2+y^2+z^2)\} $$
The induced metric in polar coordinates is given by 
\begin{equation}
    g^M = (1-\alpha^2r^2)dr^2
+
r^2 \bigl(d\theta^2 + \sin^2\!\theta\,d\phi^2\bigr).
\end{equation}
where $r^2=x^2+y^2+z^2$, then $M$ is spacelike in the region $r<\frac{1}{\alpha}$.  The second fundamental form of $M$ in $\mathbb{R}^{3,1} $ is given by 
\begin{equation}
    k=\frac{\alpha}{\sqrt{1-\alpha^2r^2 }} \,\delta  
\end{equation}
where $\delta$ is the Euclidean metric. With this we can also calculate
\begin{equation}
    \tr k= \frac{\alpha (3-2\alpha^2r^2)}{(1-\alpha^2 r^2 )^\frac{3}{2}}\quad \text{and} \quad |k|^2= \frac{\alpha^2 (3-4\alpha^2 r^2 +2\alpha^4 r^4)}{(1-\alpha^2 r^2)^3}
\end{equation}
     Again, we will consider spheres of constant radius $\Sigma_r= \{x^2+y^2+z^2=r^2 \}$, these surfaces are round spheres of area $|\Sigma|=4\pi r^2$. We will see that they are Hawking surfaces with vanishing Hawking energy. The outward normal of $\Sigma_r$ is given by $\nu = \frac{1}{\sqrt{1-\alpha^2r^2 }} \partial_r$ and with this we can calculate 
\begin{equation}
    H(\Sigma_r)= \frac{2}{r\sqrt{1-\alpha^2r^2 }},  \quad P(\Sigma_r)= \frac{2\alpha }{\sqrt{1-\alpha^2r^2 }}. \quad \text{and} \quad  |\mathring{B}|=0
\end{equation}
Then $\Sigma_r$ are constant mean curvature (CMC) surfaces and also spacetime constant mean curvature (STCMC) surfaces, also they have vanishing Hawking energy.  Now we will see that $\Sigma_r$ is a Hawking surface. One can calculate $\mathrm{Ric}^M(\nu, \nu)=-\frac{2\alpha^2 }{(1-\alpha^2 r^2 )^2} $ and 
\begin{equation}
    \nabla_\nu \tr k = \frac{\alpha^3r(5-2\alpha^2 r^2)}{(1-\alpha^2r^2)^3}, \quad (\nabla_\nu k)(\nu,\nu)= \frac{3\alpha^3r}{(1-\alpha^2r^2)^3}, \quad \diver_\Sigma (k(\cdot, \nu))=0 .
\end{equation}
Then, using the results of before and that $P$ is constant on $\Sigma_r$, the  equation characterizing the Hawking surfaces reduces to
\begin{equation}
  0=  \lambda H   + H \mathrm{Ric}^M(\nu, \nu)+P( \nabla_\nu \tr k - \nabla_\nu k(\nu,\nu ))  +\frac{1}{2}H P^2 =  \lambda H.
\end{equation}
Then  $\Sigma_r$ is a Hawking surface for $\lambda=0$ (and a Willmore surface for $\lambda =\frac{2\alpha^2 }{(1-\alpha^2 r^2 )^2}$ ). Finally, one can compute that  
\begin{equation}
    f= \left( \frac{P}{H}\right)^2|k|^2+ \frac{1}{2 }(\tr k)^2  - \frac{1}{2}|k|^2  - \frac{3}{4} P^2 -\frac{P}{H}( \nabla_\nu \tr k - \nabla_\nu k(\nu,\nu ))= \left( \frac{P}{H}\right)^2|k|^2 >0 
\end{equation}
Then $\int_{\Sigma_r} f-\lambda d\mu= \int_{\Sigma_r} \left( \frac{P}{H}\right)^2|k|^2 d\mu >0 $, and it violates our assumption of nonnegativity. 
\end{example}
\begin{remark}
    These two examples show that Hawking surfaces need not satisfy our original integral nonnegativity hypothesis.  Indeed, in both cases  $\int_{\Sigma_r} f-\lambda d\mu= \int_{\Sigma_r} \left( \frac{P}{H}\right)^2|k|^2 d\mu >0 $ (and hence also  $\int_{\Sigma_r} f_\beta-\lambda d\mu >0 $). This demonstrates that the hypothesis $\int_{\Sigma} f-\lambda d\mu\leq 0 $ is not optimal. If, instead, one uses the modified function \(\tilde f\) introduced in (\ref{ftil}) (see Remark~\ref{remarf}), then \(\tilde f\equiv0\) on both examples and $\int_{\Sigma_r}\tilde f-\lambda\,d\mu=0$, showing that the \(\tilde f\)–condition is more appropriate and less restrictive in these particular cases.  By contrast, the rigidity condition for $\Tilde{f}_\beta$ is still violated $\int_{\Sigma_r} \Tilde{f}_\beta-\lambda d\mu >0 $.
\end{remark}
 We can also consider an $n$-dimensional initial data set $(M,g,k)$, in this case, a hypersurface $\Sigma$ is an area-constrained critical surface of the Hawking functional if it satisfies 
\begin{equation}\label{eulaghigh}
\begin{split}
   0=& \lambda H  +\Delta^\Sigma H - \frac{n-3}{2(n-1)}H^3 + H|\mathring{B}|^2+ H \mathrm{Ric}^M(\nu, \nu)+P( \nabla_\nu \tr k - \nabla_\nu k(\nu,\nu )) \\ & - 2P \diver_\Sigma (k(\cdot, \nu)) +\frac{1}{2}H P^2 - 2k (\nabla^\Sigma P, \nu ).
    \end{split}
\end{equation}
In this case, we again consider two possible generalizations to the Hawking energy 
\begin{equation}
 \mathcal{E}_{n,1}(\Sigma)=   \frac{1}{2(n-1)(n-2) \omega_{n-1}} 
\left( \frac{|\Sigma|}{\omega_{n-1}} \right)^{\frac{1}{n-1}}
\int_{\Sigma} 
\left( \mathrm{Sc}^{\Sigma} - \frac{n-2}{n-1} (H^2-P^2) \right)
d\mu
\end{equation}
and 
\begin{equation}
\mathcal{E}_{n,2}(\Sigma) =    \frac{1}{2} \left( \frac{|\Sigma|}{\omega_{n-1}} \right)^{\frac{n-2}{n-1}} \left(1 - \frac{1}{(n-1)^2 \omega_{n-1}} \left( \frac{\omega_{n-1}}{ |\Sigma|} \right)^{\frac{n-3}{n-1}} \int_{\Sigma }H^2 -P^2 d\mu \right),
\end{equation}
Then similar to Theorem \ref{rigidityhighdim} and Theorem \ref{positivity0} we have the following nonnegativity result.
\begin{theorem}\label{posidynhigh}
    Let $(M,g,k)$ be a complete $n$-dimensional initial data set (with $n\geq 3$) satisfying the dominant energy condition.   Let $\Sigma $ be a Hawking surface with positive mean curvature
    and  parameter $\lambda$, and let
  $$f:= \left( \frac{P}{H}\right)^2|k|^2+ \frac{1}{2 }(\tr k)^2 -  \frac{1}{2} |k|^2 -|J|   - \frac{n}{2(n-1)} P^2- \frac{P}{H}( \nabla_\nu \tr k - \nabla_\nu k(\nu,\nu )) -\frac{1}{2} |\mathring{B}|^2$$
\begin{itemize}
    \item[$i)$]  If $
         \lambda  \geq \frac{1}{|\Sigma|}\int_{\Sigma}  f+ \frac{n-3}{2(n-2)}  \mathrm{Sc}^{\Sigma} d\mu$, 
       then  $\mathcal{E}_{n,1}(\Sigma)\geq 0$.
    \item[$ii)$]  If, $
         \lambda  \geq  \frac{1}{|\Sigma|}\int_{\Sigma} f+
 \frac{\mathrm{Sc}^{\Sigma}}{2} d\mu  -\frac{n-1}{2} \left( \frac{\omega_{n-1}}{ |\Sigma|} \right)^{\frac{2}{n-1}},$
       then  $\mathcal{E}_{n,2}(\Sigma)\geq 0$. 
\end{itemize}
Furthermore, if one of the inequalities of $\lambda$ is strict then the respective Hawking energy is positive.
\end{theorem}
\begin{proof}
    The proof is a direct combination of Theorems \ref{rigidityhighdim} and \ref{positivity0}.
\end{proof}
Note that when $n=3$ or $k=0$, the result reduces to Theorem \ref{rigidityhighdim} or Theorem \ref{positivity0} respectively.

In summary, we have seen that on critical surfaces and under the dominant energy condition,  the Hawking energy is nonnegative. In fact,  if it vanishes on such a surface, the enclosed region must be flat—directly tying the energy measure to spacetime curvature and confirming its ability to distinguish flat from curved geometries.

The extension of these properties to the general case, where the second fundamental form $k$ is nonzero, represents a major advancement. Unlike the time-symmetric case, the general case encompasses dynamical effects, making these results more broadly applicable to realistic astrophysical scenarios, such as binary mergers or gravitational wave emissions. The inclusion of dynamical contributions further enhances the Hawking energy's relevance in describing localized gravitational phenomena.

Despite these promising results, certain technical conditions imposed throughout this analysis may not be optimal or physically motivated. In particular, conditions such as
$$\int_\Sigma f_\beta -\lambda \, d\mu \leq 0$$
were introduced primarily to facilitate mathematical treatment, but it remains unclear whether they represent the most physically natural constraints for quasi-local energy formulations. A refined condition could potentially lead to stronger rigidity results.

Lastly, an important aspect to consider is that the definition of Hawking surfaces is inherently gauge-dependent. Since these surfaces are defined within a given spacelike hypersurface, any attempt to define them without a hypersurface would require selecting a preferred spacelike normal direction for variation, introducing an additional ambiguity in their construction. This gauge dependence could impact their role in general quasi-local energy formulations.

Overall, these results mark significant progress in establishing the Hawking energy as a viable quasi-local energy measure. However, further refinements in its formulation and conditions are necessary. Nevertheless, Hawking surfaces currently provide a promising framework for evaluating the Hawking energy and could prove highly valuable in numerical simulations, particularly in evolution problems that are studied on a given spacelike initial data set.

\vspace{1.5 cm}

 \paragraph*{\emph{Acknowledgements.}} The Author would like to thank Jan Metzger and Marc Mars for the helpful comments about this work. This research received partial support from the DFG as part of the SPP 2026 Geometry at infinity, project ME 3816/3-1.



\bibliographystyle{amsplain}
\bibliography{Lit_new}

@article {Miao,
    AUTHOR = {Miao, Pengzi and Wang, Yaohua and Xie, Naqing},
     TITLE = {On {H}awking mass and {B}artnik mass of {CMC} surfaces},
   JOURNAL = {Math. Res. Lett.},
  FJOURNAL = {Mathematical Research Letters},
    VOLUME = {27},
      YEAR = {2020},
    NUMBER = {3},
     PAGES = {855--885},
      ISSN = {1073-2780},
   MRCLASS = {53C20 (53C42)},
  MRNUMBER = {4216572},
MRREVIEWER = {Rafael L\'{o}pez},
       DOI = {10.4310/MRL.2020.v27.n3.a12},
       URL = {https://doi.org/10.4310/MRL.2020.v27.n3.a12},
}

@incollection {Chriyau,
    AUTHOR = {Christodoulou, Demetrios and Yau, Shing-Tung},
     TITLE = {Some remarks on the quasi-local mass},
 BOOKTITLE = {Mathematics and general relativity ({S}anta {C}ruz, {CA},
              1986)},
    SERIES = {Contemp. Math.},
    VOLUME = {71},
     PAGES = {9--14},
 PUBLISHER = {Amer. Math. Soc., Providence, RI},
      YEAR = {1988},
   MRCLASS = {83C99 (58E12)},
  MRNUMBER = {954405},
MRREVIEWER = {K. P. Tod},
       DOI = {10.1090/conm/071/954405},
       URL = {https://doi.org/10.1090/conm/071/954405},
}

@article {willflat,
    AUTHOR = {Lamm, Tobias and Metzger, Jan and Schulze, Felix},
     TITLE = {Foliations of asymptotically flat manifolds by surfaces of
              {W}illmore type},
   JOURNAL = {Math. Ann.},
  FJOURNAL = {Mathematische Annalen},
    VOLUME = {350},
      YEAR = {2011},
    NUMBER = {1},
     PAGES = {1--78},
      ISSN = {0025-5831},
   MRCLASS = {53C12 (53C24)},
  MRNUMBER = {2785762},
MRREVIEWER = {Jesse Ratzkin},
       DOI = {10.1007/s00208-010-0550-2},
       URL = {https://doi.org/10.1007/s00208-010-0550-2},
}

@article {Ce,
    AUTHOR = {Metzger, Jan},
     TITLE = {Foliations of asymptotically flat 3-manifolds by 2-surfaces of
              prescribed mean curvature},
   JOURNAL = {J. Differential Geom.},
  FJOURNAL = {Journal of Differential Geometry},
    VOLUME = {77},
      YEAR = {2007},
    NUMBER = {2},
     PAGES = {201--236},
      ISSN = {0022-040X},
   MRCLASS = {53C12 (53C20)},
  MRNUMBER = {2355784},
MRREVIEWER = {John Urbas},
       URL = {http://projecteuclid.org/euclid.jdg/1191860394},
}

@article {Eich,
    AUTHOR = {Eichmair, Michael and Metzger, Jan},
     TITLE = {Unique isoperimetric foliations of asymptotically flat
              manifolds in all dimensions},
   JOURNAL = {Invent. Math.},
  FJOURNAL = {Inventiones Mathematicae},
    VOLUME = {194},
      YEAR = {2013},
    NUMBER = {3},
     PAGES = {591--630},
      ISSN = {0020-9910},
   MRCLASS = {53C12 (58E30 83C75)},
  MRNUMBER = {3127063},
MRREVIEWER = {Kotik K. Lee},
       DOI = {10.1007/s00222-013-0452-5},
       URL = {https://doi.org/10.1007/s00222-013-0452-5},
}

@article{eichmair2013large,
  title={Large isoperimetric surfaces in initial data sets},
  author={Eichmair, Michael and Metzger, Jan},
  journal={Journal of Differential Geometry},
  volume={94},
  number={1},
  pages={159--186},
  year={2013},
  publisher={Lehigh University}
}

@phdthesis{PenuelaThesis,
  author       = {Peñuela Díaz, Alejandro},
  title        = {Geometrically defined surfaces in general
relativity and their relation to physical
invariants},
  school       = {Universit{\"a}t Potsdam},
  year         = {2025},  
}

@article{Alex,
title = {Concentration of small Hawking type surfaces},
journal = {Differential Geometry and its Applications},
volume = {85},
pages = {101927},
year = {2022},
issn = {0926-2245},
doi = {https://doi.org/10.1016/j.difgeo.2022.101927},
url = {https://www.sciencedirect.com/science/article/pii/S0926224522000808},
author={Friedrich, Alexander},
}

@article {Living,
    AUTHOR = {Szabados,   László Benő},
     TITLE = {Quasi-Local Energy-Momentum and Angular Momentum in GR},
   JOURNAL = {Living Rev. Relativity},
    VOLUME = {7},
      YEAR = {2004},
    NUMBER = {4},
}

@article {Nerz2,
    AUTHOR = {Nerz, Christopher},
     TITLE = {Foliations by stable spheres with constant mean curvature for
              isolated systems without asymptotic symmetry},
   JOURNAL = {Calc. Var. Partial Differential Equations},
  FJOURNAL = {Calculus of Variations and Partial Differential Equations},
    VOLUME = {54},
      YEAR = {2015},
    NUMBER = {2},
     PAGES = {1911--1946},
      ISSN = {0944-2669},
   MRCLASS = {53C42 (53C12 53C20)},
  MRNUMBER = {3396437},
MRREVIEWER = {Adela-Gabriela Mihai},
       DOI = {10.1007/s00526-015-0849-7},
       URL = {https://doi.org/10.1007/s00526-015-0849-7},
}

@article {Thomas,
    AUTHOR = {Koerber, Thomas},
     TITLE = {The area preserving {W}illmore flow and local maximizers of
              the {H}awking mass in asymptotically {S}chwarzschild
              manifolds},
   JOURNAL = {J. Geom. Anal.},
  FJOURNAL = {Journal of Geometric Analysis},
    VOLUME = {31},
      YEAR = {2021},
    NUMBER = {4},
     PAGES = {3455--3497},
      ISSN = {1050-6926},
   MRCLASS = {53E10 (35K25 53C80)},
  MRNUMBER = {4236532},
MRREVIEWER = {Yannan Liu},
       DOI = {10.1007/s12220-020-00401-6},
       URL = {https://doi.org/10.1007/s12220-020-00401-6},
}

@article{Wang,
	doi = {10.1088/1361-6382/ab719d},
	url = {https://doi.org/10.1088/1361-6382/ab719d},
	year = 2020,
	month = {mar},
	publisher = {{IOP} Publishing},
	volume = {37},
	number = {8},
	pages = {085004},
	author = {Wang, Jinzhao},
	title = {The small sphere limit of quasilocal energy in higher dimensions along lightcone cuts},
	journal = {Classical and Quantum Gravity},
}

@article {Hawma,
    AUTHOR = {Hawking, Stephen W.},
     TITLE = {Gravitational radiation in an expanding universe},
   JOURNAL = {J. Mathematical Phys.},
  FJOURNAL = {Journal of Mathematical Physics},
    VOLUME = {9},
      YEAR = {1968},
    NUMBER = {4},
     PAGES = {598--604},
      ISSN = {0022-2488},
   MRCLASS = {83C30},
  MRNUMBER = {3960907},
       DOI = {10.1063/1.1664615},
       URL = {https://doi.org/10.1063/1.1664615},
}

@article{eichko,
  title={Large area-constrained Willmore surfaces in asymptotically Schwarzschild $3 $-manifolds},
  author={Eichmair, Michael and Koerber, Thomas},
  journal={Journal of Differential Geometry},
  volume={127},
  number={1},
  pages={105--160},
  year={2024},
  publisher={Lehigh University}
}

@article{diaz2023local,
  title={Local foliations by critical surfaces of the Hawking energy and small sphere limit},
  author={Peñuela Díaz, Alejandro},
  journal={Classical and Quantum Gravity},
  volume={40},
  number={3},
  pages={035002},
  year={2023},
  publisher={IOP Publishing}
}

@article{sun2017rigidity,
  title={Rigidity of Hawking mass for surfaces in three manifolds},
  author={Sun, Jiacheng},
  journal={Pacific Journal of Mathematics},
  volume={292},
  number={1},
  pages={257--282},
  year={2017},
  publisher={Mathematical Sciences Publishers}
}

@article{shi2019uniqueness,
  title={Uniqueness of the mean field equation and rigidity of Hawking mass},
  author={Shi, Yuguang and Sun, Jiacheng and Tian, Gang and Wei, Dongyi},
  journal={Calculus of Variations and Partial Differential Equations},
  volume={58},
  pages={1--16},
  year={2019},
  publisher={Springer}
}

@article{shi2016isoperimetric,
  title={The isoperimetric inequality on asymptotically flat manifolds with nonnegative scalar curvature},
  author={Shi, Yuguang},
  journal={International Mathematics Research Notices},
  volume={2016},
  number={22},
  pages={7038--7050},
  year={2016},
  publisher={Oxford University Press}
}

@article{shi2002positive,
  title={Positive mass theorem and the boundary behaviors of compact manifolds with nonnegative scalar curvature},
  author={Shi, Yuguang and Tam, Luen-Fai},
  journal={Journal of Differential Geometry},
  volume={62},
  number={1},
  pages={79--125},
  year={2002},
  publisher={Lehigh University}
}

@article{shi2007rigidity,
  title={Rigidity of compact manifolds and positivity of quasi-local mass},
  author={Shi, Yuguang and Tam, Luen-Fai},
  journal={Classical and Quantum Gravity},
  volume={24},
  number={9},
  pages={2357},
  year={2007},
  publisher={IOP Publishing}
}

@article{liu2003positivity,
  title={Positivity of quasilocal mass},
  author={Liu, Chiu-Chu Melissa and Yau, Shing-Tung},
  journal={Physical review letters},
  volume={90},
  number={23},
  pages={231102},
  year={2003},
  publisher={APS}
}

@article{liu2006positivity,
  title={Positivity of quasi-local mass II},
  author={Liu, Chiu-Chu Melissa and Yau, Shing-Tung},
  journal={Journal of the American Mathematical Society},
  volume={19},
  number={1},
  pages={181--204},
  year={2006}
}

@article{miao2017quasi,
  title={Quasi-local mass integrals and the total mass},
  author={Miao, Pengzi and Tam, Luen-Fai and Xie, Naqing},
  journal={The Journal of Geometric Analysis},
  volume={27},
  number={2},
  pages={1323--1354},
  year={2017},
  publisher={Springer}
}

@article{schoen2017positive,
  title={Positive scalar curvature and minimal hypersurface singularities},
  author={Schoen, Richard and Yau, Shing-Tung},
  journal={arXiv preprint arXiv:1704.05490},
  year={2017}
}

@article{lohkamp2016higher2,
  title={The higher dimensional positive mass theorem II},
  author={Lohkamp, Joachim},
  journal={arXiv preprint arXiv:1612.07505},
  year={2016}
}

@article{lohkamp2016higher1,
  title={The higher dimensional positive mass theorem I},
  author={Lohkamp, Joachim},
  journal={arXiv preprint arXiv:math/0608795},
  year={2016}
}

@article{mccormick2019charged,
  title={On the charged Riemannian Penrose inequality with charged matter},
  author={McCormick, Stephen},
  journal={Classical and quantum gravity},
  volume={37},
  number={1},
  pages={015007},
  year={2019},
  publisher={IOP Publishing}
}

@article{disconzi2012penrose,
  title={On the Penrose inequality for charged black holes},
  author={Disconzi, Marcelo M and Khuri, Marcus A},
  journal={Classical and Quantum Gravity},
  volume={29},
  number={24},
  pages={245019},
  year={2012},
  publisher={IOP Publishing}
}

@article{pacheco2023constructing,
  title={Constructing electrically charged Riemannian manifolds with minimal boundary, prescribed asymptotics, and controlled mass},
  author={Cabrera Pacheco, Armando J.  and Cederbaum, Carla and Gehring, Penelope and Diaz, Alejandro Pe{\~n}uela},
  journal={Journal of Geometry and Physics},
  volume={185},
  pages={104746},
  year={2023},
  publisher={Elsevier}
}

@article{ros1988compact,
  title={Compact hypersurfaces with constant scalar curvature and a congruence theorem},
  author={Ros, Antonio},
  journal={Journal of Differential Geometry},
  volume={27},
  number={2},
  pages={215--220},
  year={1988},
  publisher={Lehigh University}
}

@article{maximo2012hawking,
author = {Máximo, Davi and Nunes, Ivaldo},
year = {2012},
month = {06},
pages = {},
title = {Hawking mass and local rigidity of minimal two-spheres in
three-manifolds},
volume = {21},
journal = {Communications in Analysis and Geometry},
doi = {10.4310/CAG.2013.v21.n2.a6}
}

@article{barros2017hawking,
  title={Hawking mass and local rigidity of minimal surfaces in three-manifolds},
  author={Barros, A and Batista, R and Cruz, T},
  journal={Communications in Analysis and Geometry},
  volume={25},
  number={1},
  pages={1--23},
  year={2017},
  publisher={International Press of Boston}
}

@article{baltazar2023local,
author = {Baltazar, H. and Barros, Abdênago and Batista, Rondinelle},
year = {2023},
month = {08},
pages = {},
title = {A local rigidity theorem for minimal two-spheres in charged time-symmetric initial data set},
volume = {113},
journal = {Letters in Mathematical Physics},
doi = {10.1007/s11005-023-01713-8}
}

@article{sousa2023charged,
  title={Charged hawking mass and local rigidity of three-manifolds},
  author={Sousa, Paulo A and Lima, Alexandre B},
  journal={The Journal of Geometric Analysis},
  volume={33},
  number={1},
  pages={11},
  year={2023},
  publisher={Springer}
}

@article{mondino2022some,
  title={Some rigidity results for the Hawking mass and a lower bound for the Bartnik capacity},
  author={Mondino, Andrea and Templeton-Browne, Aidan},
  journal={Journal of the London Mathematical Society},
  volume={106},
  number={3},
  pages={1844--1896},
  year={2022},
  publisher={Wiley Online Library}
}

@article{bray2007generalized,
  title={Generalized inverse mean curvature flows in spacetime},
  author={Bray, Hubert and Hayward, Sean and Mars, Marc and Simon, Walter},
  journal={Communications in mathematical physics},
  volume={272},
  pages={119--138},
  year={2007},
  publisher={Springer}
}

@article{bray2011positive,
  title={On the positive mass, Penrose, an ZAS inequalities in general dimension, Surveys in Geometric Analysis and Relativity, Adv. Lect. Math.(ALM), 20},
  author={Bray, Hubert L.},
  journal={Int. Press, Somerville, MA},
  year={2011}
}

@article{ajm/1383923956,
author = {Bray, Hubert L. and  Jauregui, Jeffrey L.},
title = {{A geometric theory of zero area singularities in general relativity}},
volume = {17},
journal = {Asian Journal of Mathematics},
number = {3},
publisher = {International Press of Boston},
pages = {525 -- 560},
year = {2013},
}

@article{miao2010geometric,
  title={On geometric problems related to Brown-York and Liu-Yau quasilocal mass},
  author={Miao, Pengzi and Shi, Yuguang and Tam, Luen-Fai},
  journal={Communications in mathematical physics},
  volume={298},
  number={2},
  pages={437--459},
  year={2010},
  publisher={Springer}
}

@article{murchadha2004comment,
  title={Comment on “Positivity of quasilocal mass”},
  author={Ó Murchadha,  Niall and Szabados, László Benő and Tod, Paul},
  journal={Physical review letters},
  volume={92},
  number={25},
  pages={259001},
  year={2004},
  publisher={APS}
}

@article{kijowski1997simple,
  title={A simple derivation of canonical structure and quasi-local Hamiltonians in general relativity},
  author={Kijowski, Jerzy},
  journal={General Relativity and Gravitation},
  volume={29},
  number={3},
  pages={307--343},
  year={1997},
  publisher={Springer}
}

@article{brown1993quasilocal,
  title={Quasilocal energy and conserved charges derived from the gravitational action},
  author={Brown, J David and York Jr, James W},
  journal={Physical Review D},
  volume={47},
  number={4},
  pages={1407},
  year={1993},
  publisher={APS}
}

@article{melo2024hawking,
  title={On the Hawking mass for CMC surfaces in positive curved 3-manifolds},
  author={Melo, Luiz Ricardo},
  journal={Proceedings of the American Mathematical Society},
  volume={152},
  number={12},
  pages={5373--5380},
  year={2024}
}

@article{hang2009rigidity,
  title={Rigidity theorems for compact manifolds with boundary and positive Ricci curvature},
  author={Hang, Fengbo and Wang, Xiaodong},
  journal={Journal of Geometric Analysis},
  volume={19},
  pages={628--642},
  year={2009},
  publisher={Springer}
}

@article{nirenberg1953weyl,
  title={The Weyl and Minkowski problems in differential geometry in the large},
  author={Nirenberg, Louis},
  journal={Communications on pure and applied mathematics},
  volume={6},
  number={3},
  pages={337--394},
  year={1953},
  publisher={Wiley Online Library}
}

@book{pogorelov1972extrinsic,
  title={Extrinsic geometry of convex surfaces},
  author={Pogorelov, Viktor Andreevich },
  volume={35},
  year={1972},
  publisher={American Mathematical Soc.}
}

@article{lee2025modified,
  title={Modified Hawking mass and rigidity of three-manifolds with boundary},
  author={Lee, Jihyeon and Lee, Sanghun},
  journal={arXiv preprint arXiv:2505.08301},
  year={2025}
}

@article{leandro2025sharp,
  title={Sharp lower bound for the charged Hawking mass in the electrostatic space},
  author={Leandro, Benedito and Sabo, Guilherme},
  journal={arXiv preprint arXiv:2507.03353},
  year={2025}
}

@article{min1998scalar,
  title={Scalar curvature rigidity of certain symmetric spaces},
  author={Min-Oo, Maung},
  journal={Geometry, topology, and dynamics (Montreal, 1995)},
  volume={127137},
  year={1998}
}

@article{brendle2011deformations,
  title={Deformations of the hemisphere that increase scalar curvature},
  author={Brendle, Simon and Marques, Fernando C and Neves, Andre},
  journal={Inventiones mathematicae},
  volume={185},
  number={1},
  pages={175--197},
  year={2011},
  publisher={Springer}
}

@article{shi2009behavior,
  title={On the behavior of quasi-local mass at the infinity along nearly round surfaces},
  author={Shi, Yuguang and Wang, Guofang and Wu, Jie},
  journal={Annals of Global Analysis and Geometry},
  volume={36},
  number={4},
  pages={419--441},
  year={2009},
  publisher={Springer}
}

@article{wei2021minimizers,
  title={On the minimizers of curvature functionals in asymptotically flat manifolds},
  author={Wei, Guodong},
  journal={The Journal of Geometric Analysis},
  volume={31},
  number={6},
  pages={5837--5853},
  year={2021},
  publisher={Springer}
}

@article{guan2009quermassintegral,
  title={The quermassintegral inequalities for k-convex starshaped domains},
  author={Guan, Pengfei and Li, Junfang},
  journal={Advances in Mathematics},
  volume={221},
  number={5},
  pages={1725--1732},
  year={2009},
  publisher={Elsevier}
}

@article{penuelafol,
  author  = {Peñuela Díaz, Alejandro}, 
  title   = {Foliations by critical surfaces of the Hawking energy in asymptotically flat initial data sets},
  journal = {Communications in Analysis and Geometry},
  volume  = {34}, 
  number  = {2}, 
  pages   = {505--570},
  year    = {2026}, 
  doi     = {10.4310/CAG.260531221657} 
}
\end{document}